\documentclass{article}

\usepackage{fullpage}

\usepackage{amsmath,amssymb,xcolor}
\usepackage{enumitem}

\usepackage[hyperfootnotes=false,
   colorlinks,
   linkcolor={blue},
   citecolor={blue},
   urlcolor={blue}
   ]{hyperref}

\usepackage[backend=biber,style=alphabetic,maxnames=10]{biblatex}

\usepackage{xurl}

\renewbibmacro*{url+urldate}{%
  \setunit{\newline}%
  \printfield{url}%
  \setunit*{\addspace}%
  \printurldate}

\usepackage{amsthm,thmtools}
\declaretheoremstyle[qed=$\lrcorner$,bodyfont=\it]{it}
\declaretheoremstyle[qed=$\lrcorner$,bodyfont=\rm]{rm}

\declaretheorem[style=it,numberwithin=subsection]{lemma}
\def\thelemma{\thesubsection.\Alph{lemma}}

\declaretheorem[style=it,numberlike=lemma]{corollary}
\declaretheorem[style=it,numberlike=lemma]{theorem}
\declaretheorem[style=it,numberlike=lemma]{proposition}

\declaretheorem[style=it]{remark}
\declaretheorem[style=it,numbered=no]{definition}
\declaretheorem[style=rm,numberlike=lemma]{question}

\declaretheorem[style=it,name=Theorem]{theoremA}

\declaretheorem[style=it,name=Example,numberlike=theoremA]{exampleA}
\declaretheorem[style=it,name=Conjecture,numberlike=theoremA]{conjecture}
\declaretheorem[style=it,name=Open~Problem,numberlike=theoremA]{problem}

\makeatletter

\newcommand{\XYAB}{XY\mkern -4.5mu AB}
\newcommand{\unskipp}[1][-6ex]{\mbox{}\vspace{#1}\par}

\let\tmp=\phi \let\phi=\varphi \let\varphi=\tmp
\let\tmp=\epsilon \let\epsilon=\varepsilon \let\varepsilon=\tmp

\usepackage{bm}

\def\mkmathletter#1#2{%
    \expandafter\gdef\csname#1#2\endcsname%
           {\ensuremath{\csname math#2\endcsname{#1}}}%
}
\edef\letters{a,b,c,d,e,f,g,h,i,j,k,l,m,n,o,p,q,r,s,t,u,v,w,x,y,z}
\edef\Letters{A,B,C,D,E,F,G,H,I,J,K,L,M,N,O,P,Q,R,S,T,U,V,W,X,Y,Z}
\edef\Lletters{\letters,\Letters}
\def\mkmathletters#1#2{%
        \@for\ltr:={#2}\do{\expandafter\mkmathletter\ltr{#1}}%
}
\mkmathletters{bb}{\Letters}    
\mkmathletters{bf}{\Lletters}   
\mkmathletters{rm}{\Lletters}   
\mkmathletters{sf}{\Lletters}   
\mkmathletters{bar}{\Lletters}  
\def\st{\; {\bm :}\; }
\def\set#1{\left\{#1\right\}}
\def\vect#1{\left(\!\!\!\begin{array}{c}#1\end{array}\!\!\!\right)}
\let\dmo=\DeclareMathOperator
\dmo\ing{\mathsf{ing}}
\let\L=\underfined
\dmo\L{\mathrm{L}}
\dmo\gr{\mathrm{Gr}}
\dmo\supp{\mathrm{supp}}
\makeatother
\def\exclam#1{\ifmmode\relax\else\marginpar{\textcolor{#1}{\fbox{!!!}}}\fi}

\begin{document}
\title{Spectral Conditions for the Ingleton Inequality}
\author{Rostislav Matveev and Andrei Romashchenko}

\maketitle
\begin{abstract}
  The Ingleton inequality is a classical linear information inequality
  that holds for rank functions of representable matroids but fails
  for general entropic vectors.
  Understanding the extent of its possible violations
  has been a longstanding problem in information theory.
  In this article, we show that for a class of jointly distributed
  random variables $(X,Y,A,B)$ the Ingleton inequality holds up to an
  additive error of order $\log H(\XYAB)$,
  even though the mutual information between $X$ and $Y$ is strongly
  non-extractable.
  Contrary to common intuition, strongly non-extractable mutual
  information does not lead to large violations of the Ingleton
  inequality in this setting.
  More precisely, we consider pairs $(X,Y)$ that are uniformly
  distributed on their joint support and whose associated biregular
  bipartite graph is an expander.
  For all auxiliary random variables $A$ and $B$ jointly distributed
  with $(X,Y)$, we establish a lower bound on the Ingleton quantity
  \[
    I(X; Y \mid A) + I(X; Y \mid B) + I(A; B) - I(X; Y)
  \]
  in terms of the spectral parameters of the underlying graph.
  Our proof combines the expander mixing lemma with a partitioning
  technique for finite sets.
\end{abstract}


\section{Introduction}

\subsection{Ingleton inequality}
The Ingleton inequality was originally introduced as an inequality
satisfied by the rank function of any representable matroid.  It was
used by A.~W.~Ingleton as a necessary condition for the
representability of a matroid over a field~\cite{ingleton}.  Since the
1990s, its entropic analog has been an active topic of research in
information theory.  Using standard information-theoretic notation,
this inequality can be written as
\begin{equation}
  \label{eq:ingleton}
  I(X; Y|A) + I(X; Y|B) + I(A; B) - I(X; Y) \ge 0.
\end{equation}
It is known that the entropic Ingleton inequality is not universally
valid: there exist joint distributions $(X,Y,A,B)$ that violate
\eqref{eq:ingleton}; 
see~\cite{matuvs1995conditional,csirmaz1996dealer,zhang1998characterization,
  matu1999conditional,hammer2000inequalities}.
The problem of quantifying the maximal possible violation
has been studied extensively; 
see, e.g.,
\cite{dougherty2011non,matuvs2016entropy,boston2020violations}.
A complete answer would amount to characterizing
achievable entropy profiles for $4$-tuples of random variables.

The entropy version of the Ingleton inequality has been used, among
other contexts, to derive bounds for linear network
coding~\cite{dougherty2005insufficiency,dougherty2007networks,li2003linear}
and bounds on the information ratio of linear secret sharing
schemes~\cite{padro2013finding,beimel2008matroids,csirmaz2009impossibility}.
A natural question is to characterize the class of
distributions for which the Ingleton inequality holds, either exactly
or up to an additive error.

\subsection{Perturbed Ingleton inequalities}
It is known that \eqref{eq:ingleton} holds for all distributions
for which the mutual information between $X$ and $Y$ is
\emph{extractable}, i.e., the mutual information of the pair $(X,Y)$
coincides with its common information in the sense of G\'acs and
K\"orner~\cite{gacs1973common}. More explicitly, we say that the
mutual information between $X$ and $Y$ is (totally) extractable if
there exists a random variable $W$ such that
\begin{equation}
  \label{eq:common-information}
  I(X; Y|W)=0, \qquad H(W|X) = 0, \qquad H(W|Y) = 0 .
\end{equation}
Under these assumptions, the proof of~\eqref{eq:ingleton} closely
parallels the original proof of the Ingleton inequality for rank
functions of representable matroids.%
\footnote{Recently, B\'erczi et al. gave a different proof of the
  Ingleton inequality for representable matroids. Their method links
  the inequality to the existence of universal tensor products of
  matroids. This approach extends to a larger class of matroids,
  including skew-representable ones~\cite{berczi2026interaction}.%
} %
In the information--theoretic setting, the G\'acs--K\"orner common
information (the random variable $W$ in \eqref{eq:common-information})
plays the role of intersection of two linear subspaces,
see~\cite{hammer2000inequalities,dougherty2009linear} for a more
detailed discussion.

The standard proof of the implication
\( \eqref{eq:common-information} \Longrightarrow \eqref{eq:ingleton}
\) gives a robust approximate version of the Ingleton inequality, in
which the zero on the right-hand side of~\eqref{eq:ingleton} is
replaced by a combination of the entropy quantities appearing
in~\eqref{eq:common-information};  see
Lemma~\ref{p:shannon-ineq}\ref{i:ing-soft}.  Tighter robust versions
can be obtained using non-Shannon-type inequalities.  In particular,
we can apply an inequality proved in~\cite{makarychev2002new}.
Namely, for every quintuple of jointly distributed random variables
$(X,Y,A,B,W)$, the following inequality holds:
\begin{equation}
  \label{eq:mmrv}
    I(X; Y | A)
    + I(X; Y| B) + I(A; B) - I(X; Y) 
    \ge - \bigl( I(X; Y | W) + I(Y; W | X) + I(X; W|Y) \bigr).
\end{equation}
This inequality generalizes the first non-Shannon-type information
inequality of Zhang and Yeung~\cite{zhang1998characterization}.
We refer to it as the MMRV inequality.

The MMRV inequality immediately implies 
\begin{equation*}
  \label{eq:mmrv-weak}
    I(X; Y | A)
    + I(X; Y| B) + I(A; B) - I(X; Y)
   \ge - \bigl( I(X; Y | W) + H(W | X) + H(W|Y) \bigr).
\end{equation*}
Thus, a small violations of the
equalities in~\eqref{eq:common-information} lead only to a small
violation of the Ingleton inequality~\eqref{eq:ingleton}. For this
reason, it is natural to expect that strong violations of the Ingleton
inequality require distributions for which the mutual information
between $X$ and $Y$ is far from extractable: namely, distributions for
which, for every auxiliary random variable $W$, the equalities
in~\eqref{eq:common-information} are violated as strongly as possible.

In this article, we show that this intuition seems to be
misleading.  Specifically, we prove that for a natural class of pairs
$(X,Y)$ whose mutual information appears to be%
\footnote{A precise formulation is given below as
Conjecture~\ref{co:nonextractable}. A proof of this conjecture is
the subject of our on-going research and is currently being prepared for
publication.} %
strongly non-extractable (and for arbitrary extending random variables
$A$ and $B$) the Ingleton inequality holds up to additive error of
order $\log H(\XYAB)$.

\subsection{Our setup and the main result}

We focus on pairs of random variables $(X,Y)$ that are regular in the
following sense: $(X,Y)$ is uniformly distributed on its joint
support, and the marginals $X$ and $Y$ are uniformly distributed on
their respective supports. The properties of such a distribution are
therefore fully determined by the combinatorial structure of the
support. It is convenient to represent this support as a biregular
bipartite graph $\Gsf$, which we call the \emph{support} of the pair
$(X,Y)$.  Specifically, let $\Xsf$ and $\Ysf$ denote the supports of
$X$ and $Y$, respectively, and let $\Esf \subseteq \Xsf \times \Ysf$
be the support of the joint distribution. Sampling $(X,Y)$ then
amounts to choosing an edge uniformly at random from the biregular
bipartite graph $\Gsf=(\Xsf\sqcup\Ysf,\Esf)$, with $X$ and
$Y$ being the left and right endpoints of this edge. For such
distributions, entropies are given by the support sizes:
\[
  H(X) = \log |\Xsf|,
  \qquad
  H(Y) = \log |\Ysf|,
  \qquad
  H(X,Y) = \log |\Esf|.
\]
Consequently,
\[
  I(X; Y) = \log \frac{|\Xsf| \cdot |\Ysf|}{|\Esf|}.
\]

In this article we prove the following theorem.

\begin{theoremA}[Main result, see Theorem~\ref{p:main} in Section~\ref{s:main}]
  \label{th:main}\ \\
  Let $\epsilon_{0}>0$ be a positive constant. 
  Suppose $(X,Y)$ is ob\-tain\-ed by choosing a uniformly random edge in a
  biregular bipartite graph $G=(\Xsf\sqcup\Ysf,\Esf)$ whose largest
  and second-largest eigenvalues are $\lambda_1$ and $\lambda_2$,
  respectively, and such that $I(X; Y)>\epsilon_0$. 
  Then for all random variables $A,B$ jointly
  distributed with $(X,Y)$,
  \begin{equation*}
    I(X ;  Y | A) \!+\! I(X ;   Y | B)\! +\! I(A; B)\! -\! I(X ;  Y)
    \ge \log\!\left(\!\frac{\lambda_1}{\lambda_2^2}\!\right)
    \!-\! O\bigl(\log H(X,Y,A,B)\bigr)
  \end{equation*}
  where the $O(\cdot)$-summand implicitly  depends only on the chosen
  threshold value $\epsilon_{0}$.
\end{theoremA}

The conclusion of the theorem is especially strong when $\lambda_2$
is small. If the mutual information $I(X; Y)$ is bounded by a
positive constant $\epsilon_{0}$ from below, or equivalently, if the
density of the support graph $\Gsf$ is separated from 1, then we
have the following bounds (Proposition~\ref{p:alonlike})
\begin{equation*}
  2\log\lambda_{1}\geq 2\log\lambda_{2}\geq\log\lambda_{1}-C(\epsilon_{0})
\end{equation*}
Thus, for balanced expander families, that is, growing families of
biregular bipartite graphs, such that the sizes of the left- and right
parts are the same and such that the lower bound on the
$\log\lambda_{2}$ is attained, the conclusion of the theorem reads
\begin{equation*}
  I(X; Y | A) \!+\! I(X; Y | B)\! +\! I(A; B)\! -\!
  I(X; Y)
  \geq
  - O\bigl(\log H(X,Y,A,B)\bigr)
\end{equation*}

\begin{remark}
  The conclusion of Theorem~\ref{th:main} would be false without the
  residue term $O\bigl(\log H(X,Y,A,B)\bigr)$ on the right-hand side.
  Indeed, taking $A=X$ and $B=Y$ makes the left-hand side equal to
  $0$. At the same time, there exist graphs $G$ with
  $\log\!\left(\frac{\lambda_1}{\lambda_2^2}\right) > 0$;  see, for
  example, the family of graphs in Example~\ref{ex:projective-flags}
  below.  Thus, the term
  $ \log\!\left(\!\frac{\lambda_1}{\lambda_2^2}\!\right)$ in the
  right-hand side of the inequality in Theorem~\ref{th:main} can be
  strictly positive.
  
  We do not know whether $O\bigl(\log H(X,Y,A,B)\bigr)$ in this
    inequality can be improved to $O\bigl(\log H(X,Y)\bigr)$ or
  $O(1)$. See Sections~\ref{s:examples} and~\ref{s:discussion} for the
  detailed discussion and examples.
\end{remark}

We emphasize that in Theorem~\ref{th:main} the regularity assumption
(uniformity on the support) is imposed only on the pair $(X,Y)$;  no
regularity assumptions are required on the auxiliary random variables
$A$ and $B$.

\subsection{Relation to the extractable information}
It is expected that the mutual information is not extractable for
pairs $(X,Y)$ supported on expanders.
Technically, this means that at least one of the equalities in
\eqref{eq:common-information} must be strongly violated. In framework
of Kolmogorov complexity this phenomenon was observed in
\cite[theorem~2 and theorem~3~(ii)]{caillat2024common-arxiv}.%
\footnote{%
  For $(X,Y)$ from Example~\ref{ex:projective-flags}, discussed below,
  the non-extractability of mutual information in terms of Kolmogorov
  complexity was first established by different methods in
  \cite[Theorem 3]{muchnik1998common} and~\cite[Theorems 4 and
  8]{chernov2002upper}.} %
The following conjecture, which formalizes this claim in the context
of Shannon's information theory, asserts that the quantities
in~\eqref{eq:common-information} cannot all be simultaneously small,
but instead satisfy a specific nontrivial trade-off.  Establishing
this conjecture is the subject of our ongoing research.
\begin{conjecture}
  \label{co:nonextractable}
  There exists a universal constant $C>0$, such that for every
  uniform on the support pair $(X,Y)$ of random variables and any
  jointly distributed $W$ the
  following inequality holds
  \[
    I(X; Y|W) + I(W; X|Y) + I(W; Y|X)
    \geq
    \min\set{I(X\!; \!Y),2\log\frac{\lambda_{1}}{\lambda_{2}}}-C\cdot\log H(XY)
  \]
  where $\lambda_{1}$ and $\lambda_{2}$ are the largest- and,
  respectively, the second largest eigenvalues of the graph supporting
  $(X,Y)$.
\end{conjecture}

Assuming the conjecture holds true, the MMRV inequality applied to a
pair supported on a balanced expander becomes a trivial (Shannon-type)
inequality and cannot imply even an approximate form of the Ingleton
inequality, whereas Theorem~\ref{th:main} provides a strong lower
bound on the Ingleton expression.

On the other hand, if mutual information is extractable, i.e., if the
equalities~\eqref{eq:common-information} hold, then the MMRV
inequality~\eqref{eq:mmrv} is strongest possible (the three
  terms in the right-hand side of \eqref{eq:mmrv} vanish) and implies
the Ingleton inequality~\eqref{eq:ingleton}.  In this case,
$\lambda_2$ reaches its upper bound $\lambda_1$, and the conclusion of
Theorem~\ref{th:main} reduces to a Shannon-type inequality, see
Section~\ref{s:spectral-mmrv} for details.

Thus, the classical approach (and its refinement via the MMRV
inequality) and Theorem~\ref{th:main} provide complementary techniques
for proving approximate forms of the Ingleton inequality, as they are
effective in different regimes.

\subsection{Examples}
\begin{exampleA}[See Section~\ref{s:linear}]
  \label{ex:projective-flags}
  Fix a finite field $\mathbb{F}$ and consider the projective plane
  over $\mathbb{F}$.  Let $(X,Y)$ be a random incidence in this plane,
  where $X$ is a uniformly random point and $Y$ is a uniformly random
  line incident to $X$.  A direct calculation shows that
  \[
    I(X; Y) = \log |\mathbb{F}| + O(1)
  \]
  It is known (see Section~\ref{s:linear}) that
  the incidence graph of a finite projective plane satisfies
  $\lambda_2 \le \sqrt{\lambda_1}$.  Therefore, the inequality in
  Theorem~\ref{th:main} rewrites to
  \[
    I(X; Y | A) + I(X; Y | B) + I(A; B) - I(X; Y)
    \ge - O\bigl(\log H(X,Y,A,B)\bigr).
  \]
\end{exampleA}

Previous example is symmetric under exchanging $X$ and $Y$,
due to projective duality.  Theorem~\ref{th:main} can be also applied
to asymmetric pairs, as in the next example.
\begin{exampleA}[See Section~\ref{s:algebraic}]
  Fix a finite field $\mathbb{F}$. Let $\Fbb^{(2)}[t]$ be the space of
  polynomials of degree at most $2$ over this field. We define the
  left and the right vertex sets of the biregular bipartite graph
  $\Fsf=(\Xsf\sqcup\Ysf,\Esf)$ by
  \[
    \Xsf:=\Fbb^{2},\quad
    \Ysf:=\Fbb^{(2)}[t],\quad
    \Esf
    =
    \set{\bm{(}\vect{x\\y},p\bm{)}\in\Fbb^{2}\times\Fbb^{(2)}[t]
      \st y=p(x)}
  \]
  Let $\Fbf=(X,Y)$ be the pair supported on $\Fsf$, i.e., the
  supports of $X$, $Y$ and the joint $XY$ are $\Xsf$, $\Ysf$ and
  $\Esf$, respectively.  For this distribution we have
  $I(X; Y)=\log |\mathbb{F}|$.  Computing the spectrum of $\Fsf$ and
  applying Theorem~\ref{th:main} gives
  \[
    I(X; Y | A) + I(X; Y | B) + I(A; B) - I(X; Y) 
     \geq -\frac{\log |\mathbb{F}|}{2} -O(\log H(\XYAB))
  \]  
\end{exampleA}
  
\subsection{Proof technique}
The proofs in this article rely on two main ingredients: the expander
mixing lemma (in the form adapted for bipartite graphs presented
in~\cite{evra2015mixing}), and combinatorial partitioning results for
multidimensional finite sets from~\cite{alon2007partitioning}.

We conclude the introduction with a brief outline of the proof of our
main result (Theorem~\ref{th:main}).  After introducing some
preparatory material in Section~\ref{s:prelim}, we first establish our
result in a restrictive setting in which not only the pair $(X,Y)$ but
the entire quadruple $(X,Y,A,B)$ is uniformly distributed on the
support (or, more generally, \emph{approximately uniform} under a
suitable notion of approximate uniformity). This part of the proof
relies on the expander mixing lemma combined with conventional
information-theoretic inequalities, and is presented in
Section~\ref{s:ing-q-unif}. In Section~\ref{s:split}, we apply a
partitioning result of Alon et al.~\cite{alon2007partitioning} to
decompose the distribution of $(X,Y,A,B)$ into a small number of
approximately uniform components, thereby preparing for the second
step, in which we treat the general case where the joint distribution
of $(X,Y,A,B)$ need not be uniform;  this is done in
Section~\ref{s:main}. In Section~\ref{s:examples}, we provide a series
of examples where the main theorem is applicable and yields strong
conclusions. Finally, in Section~\ref{s:discussion} we comment on
possible generalizations and open questions.  For the reader's
convenience, Section~\ref{s:shannon} contains the list of
Shannon-type inequalities used in the article, together with their
justifications.

\section{Preliminaries}\label{s:prelim}
\subsection{Notations}\label{s:notations}
We use the notation $[n]:=\set{0,1,\dots,n-1}$ and $2^{S}$ for
the power set of a set $S$.
We denote by $\exp(x)$ the exponential function $e^x$.
We use $O(\cdot)$-notation in the
asymptotic regime, where the argument of $O(\cdot)$ is large; 
we assume $0\leq O(x)\leq C_{1}\cdot x + C_{2}$, where $C_{i}$ are some
non-negative constants depending on the context.

All random variables in this article take values in finite
alphabets. We will use the following syntactic-semantic rule: for
random variables $X,A,U,\dots$ we denote their alphabets by the same
letter in san-serif, that is $\Xsf,\Asf,\Usf,\dots$, respectively.  We
say that $x\in\Xsf$ is an \emph{atom} of $X$, if $\Pbb[X=x]>0$.  We
denote by $\sharp(X)$ the number of atoms, that is the cardinality of
the support of the distribution of $X$.

For a tuple of jointly distributed random variables
$\Xbf=(X_{i}\st i=0,\dots,n-1)$ and a set of indices $I\in 2^{[n]}$ we
denote by $X_{I}$ a joint random variable of the subcollection
$(X_{i}\st i\in I)$. We stress that $X_{I}$ is a single random
variable, where the marginalization structures of the
collection are forgotten. For a pair $(X,Y)$ we denote
their joint (forgetting the margins) by concatenating the letters,
$XY$. Similar notation is used for longer tuples.

We stress here the point, that the tuple of random variables
$\Xbf=(X_{i}\st i\in[n])$ and the joint random variable $X_{[n]}$ are
distinct objects in our setup. We denote the tuple as a comma
separated list, while the joint is denoted either by concatenating the
corresponding letters or by using subset of indices, as
in $X_{I}$.

We say that a tuple $(X_{i},Y_{j}\st i=0,\dots,n; \,j=0,\dots,k)$ of
jointly distributed random variables is an \emph{extension} of
$(X_{i}\st i=0,\dots,n)$ or, equivalently, $(Y_{j}\st j=0,\dots,k)$
\emph{extends} $(X_{i}\st i=0,\dots,n)$.

Besides the standard notations $H(X)$, $H(X|Y)$, $I(X; Y)$, $I(X; Y|Z)$ for
(conditional) entropy and mutual information, we will use the
following notations
\begin{align*}
  &\ing(X,Y,A,B):=I(X; Y|A)+I(X; Y|B)+I(A; B)-I(X; Y)\\
  &\ing(X,Y,A,B|U):= I(X; Y|AU)+I(X; Y|BU)+I(A; B|U)-I(X; Y|U)\\
  &\phantom{\ing(X,Y,A,B|U):}=\ing(XU,YU,AU,BU) \\
  &\L(X,Y):=\frac12H(X) + \frac12H(Y)-I(X; Y)
    =\frac12\big( H(X|Y) + H(Y|X)\big)\\
  &\L(X,Y|Z)\!:=\!\frac12H(X|Z)\! +\! \frac12H(Y|Z)\!-\!I(X; Y|Z)
    \!=\!\frac12\big( H(X|YZ)\! +\! H(Y|XZ)\big)
\end{align*}
Observe that $\L(X,Y)$ satisfies the usual properties of a metric
(positivity, symmetry, the triangle inequality), see e.g.,
\cite[exercise~2.9]{cover1999elements}.
We use natural base for logarithm through the article.

For an atom $y\in\Ysf$ in the alphabet of $Y$ we denote by $X|y$ the
conditioned random variable, in lieu of more conventional
$X|Y=y$. This should not cause a problem, since alphabets of random
variables are always explicitly given and tacitly assumed to be
disjoint.

The conditioning operator has the lowest precedence, even lower than
that of comma, e.g., 
\[
  \ing(X,Y,A,B|u)=I(X; Y|Au)+I(X; Y|Bu)+I(A; B|u)-I(X; Y|u).
\]

\subsection{Random Variables}\label{s:prelim-rvs}
Suppose tuple of jointly distributed random variables
$\Xbf:=(X_{i}\st i=0,\dots,n-1)$ take values in alphabets $\Xsf_{i}$,
$i=0,\dots,n-1$, respectively.
Entropies of partial joints of random variables in the collection
$(X_{i}\st i=0,\dots,n-1)$ satisfy the series of so-called 
\emph{Shannon inequalities}
\begin{equation*}
  I(X_{I}; X_{J}|X_{K})\geq 0 
  \ 
  \text{for non-empty $I,J\in 2^{[n]}$  and any $K\in 2^{[n]}$}
\end{equation*}
(if $I=J$, then the inequality rewrites to $ H(X_{I} | X_{K})\geq 0 $; 
if $K$ is empty, then $X_{K}$ is trivial, and the inequality is
equivalent to $ I(X_{I}; X_{J})\geq 0 $).

An inequality for entropy quantities is called \emph{Shannon-type} if
it is valid for all polymatroids, equivalently, if it can be
represented as a positive linear combination of Shannon inequalities.
  
We say that collection $\Xbf$ is \emph{uniform on the support} if for
every $I\in 2^{[n]}$ the distribution of $X_{I}$ is uniform on its
support, that is, the probability mass function takes at most two
values, zero and some positive value.
\begin{remark}
  Recall that $X_{I}$ is a single random variable, where the
  marginalization structures of the collection are forgotten.  So,
  when we say that $X_{[n]}$ is uniform on the support, this does not
  imply uniformity of $X_{I}$ for $I\neq [n]$. However, when we say
  that the collection $\Xbf$ is uniform on the support, we mean that
  for every $I \in 2^{[n]}$ the distribution of $X_I$ is uniform on
  its support.
\end{remark}

The collection $\Xbf$ is said to be \emph{regular} if
for any disjoint subsets of indices $I,J\in 2^{[n]}$ and any two atoms
$\xbf,\ybf\in\Xsf_{J}$ the supports of $X_{I}|\xbf$ and $X_{I}|\ybf$
have the same cardinality.

It is immediate that if $\Xbf$ is uniform on the support, then it is
regular. In particular, $X_{I}|\xbf$ is distributed uniformly on the
support and $H(X_{I}|X_{J})=H(X_{I}|\xbf)$ for any $I,J\in2^{[n]}$ and
any atom $\xbf\in\Xsf_{J}$.

A random variable $X$ with alphabet $\Xsf$ and distribution $p$ is
called \emph{$\delta$-uniform}, $\delta\geq1$, if
\[
  \frac{\max_{x} p(x)}{\min_{x}p(x)}\leq \delta
\]
where the extrema are taken over the set of all atoms $x\in\Xsf$.

We say that a tuple of random variables $\Xbf$ is
\emph{$\delta$-regular} if for every disjoint $I,J\in 2^{{[n]}}$
\[
  \frac{\max_{\xbf}\sharp(X_{I}|\xbf)}{\min_{\xbf}\sharp(X_{I}|\xbf)}\leq\delta, 
\]
where $\sharp(X_{I}|\xbf)$ stands for the cardinality of the support
of the distribution of $X_{I}|\xbf$ and the extrema are taken over atoms
$\xbf\in\Xsf_{J}$.

We say that a collection $\Xbf$ is \emph{$\delta$-uniform} if $\Xbf$
is $\delta$-regular and for every $I\in2^{[n]}$ random variable
$X_{I}$ is $\delta$-uniform.  In the context where the particular
value of $\delta$ is not important or not specified, we say that a
tuple is \emph{almost-uniform} or \emph{almost-regular}, respectively.

The following simple lemma is left to the reader to prove.
\begin{lemma}\label{p:unif-reg}
  For a tuple of random variables $\Xbf=(X_{i}\st i=0,\dots,n-1)$ the
  following implications hold
  \begin{enumerate}[label=(\roman*)]
  \item\label{i:unif->reg} If $X_{I}$ is $\delta$-uniform for every
    $I\in 2^{[n]}$, then the whole tuple $\Xbf$ is
    $\delta^{2}$-uniform.
  \item\label{i:reg->unif} If $\Xbf$ is $\delta$-regular, and
    $X_{[n]}$ is $\epsilon$-uniform, then the whole collection $\Xbf$
    is $\epsilon\delta$-uniform.
  \end{enumerate}
  \unskipp
\end{lemma}

We will also need the following elementary lemma.
\begin{lemma}\label{p:entropy-d-uniform}
  Let $X$ be a $\delta$-uniform random variable for some
  $\delta\geq1$  and with the support $\Xsf$.  Then
  \[
    \log|\Xsf|\geq H(X)\geq\log|\Xsf|-\log \delta.
  \]
  \unskipp
\end{lemma}
This result might be well known, but for lack of a reference and since
the proof is just a one-liner, we prove it below.
\begin{proof}
  The upper bound $H(X) \leq \log|\Xsf|$ is standard.  To prove the
  lower bound, we let \( \pi(x) := \Pbb[X = x] \) denote the
  probability distribution of $X$ and observe that
 \[
   1 = \sum_{x\in \Xsf} \pi(x)
   \ge
   |\Xsf| \cdot \min_{x\in \Xsf} \pi(x)
   \ge
   |\Xsf| \cdot \max_{x\in \Xsf} \pi(x)/\delta
 \]
 It follows 
 \[
   H(X) = \sum_x \pi(x) \log \big(\pi(x)\big)^{-1}
   \ge
   \log \left( \max_x \pi (x)  \right)^{-1} 
   \ge \log ( |\Xsf| / \delta),
 \]
 which gives the lower bound.
\end{proof}

\subsection{Graphs}\label{s:graphs}
Let $\Gsf=(\Xsf\sqcup\Ysf, \Esf)$ be a biregular bipartite graph, 
and let 
\[
  d_{1}:=d_{1}(\Gsf):=|\Esf|/|\Xsf|
  \quad\text{and}\quad
  d_{2}:=d_{2}(\Gsf):=|\Esf|/|\Ysf|
\]
stand for the left and right degrees of $\Gsf$, respectively.
We
denote by $\lambda_{i}=\lambda_{i}(\Gsf)$, $i=1,\dots$ the eigenvalues
of the adjacency matrix of $G$ (the matrix has $|\Xsf| +|\Ysf|$ rows
and columns, equal to the number of vertices in the graph) indexed in
the decreasing order and counted with multiplicity, so that
$\lambda_{1}$ and $\lambda_{2}$ are the largest and the second largest
eigenvalues. For every biregular bipartite graph $\Gsf$ holds
\[
  \lambda_{1}(\Gsf)=\frac{|\Esf|}{\sqrt{|\Xsf|\cdot|\Ysf|}}=\sqrt{d_{1}\cdot
  d_2}.
\]

For a subgraph $\Hsf$ of $\Gsf$ denote by $\Xsf_{H}$,
$\Ysf_{H}$ and $\Esf_{H}$ the left part, the right part and the edges
of $\Hsf$, respectively. We use the connection between spectral and combinatorial
properties of graphs.  A central tool is the expander mixing lemma,
originally established by Alon and Chung \cite{alon1988explicit} (see
also the modern survey of Hoory, Linial, and Wigderson
\cite{hoory2006expander} for a standard exposition).  In this article,
we will use the bipartite version of the mixing lemma in the
formulation given in \cite{evra2015mixing}.

\begin{theorem}[Bipartite Expander Mixing Lemma, see
  \cite{evra2015mixing}]\label{p:mixing}\ \\
  For any induced subgraph $\Hsf$ of a biregular bipartite graph
  $\Gsf=(\Xsf\sqcup\Ysf, \Esf)$ with the second eigenvalue
  $\lambda_{2}=\lambda_{2}(\Gsf)$ holds
  \[
    \left|
      |\Esf_{\Hsf}|
      -
      |\Esf|\cdot\frac{|\Xsf_{\Hsf}|}{|\Xsf|}\cdot\frac{|\Ysf_{\Hsf}|}{|\Ysf|}
    \right|
    \leq
    \lambda_{2} \cdot |\Xsf_{\Hsf}|^{1/2} \cdot |\Ysf_{\Hsf}|^{1/2}
  \]
  \unskipp
\end{theorem}

\begin{corollary}\label{p:mixing-log}
  Let $\Hsf$ be a (not necessarily induced) subgraph of a biregular
  bipartite graph $\Gsf=(\Xsf\sqcup\Ysf, \Esf)$ with the second
  largest eigenvalue $\lambda_{2}=\lambda_{2}(\Gsf)$. Then at least
  one of the following alternatives holds:
  \begin{enumerate}[label=(\roman*)]
  \item\label{i:boring}
    $\log |\Esf_{\Hsf}|-\log |\Xsf_{\Hsf}|-\log |\Ysf_{\Hsf}| \leq \log |\Esf| - \log |\Xsf| -
    \log |\Ysf| + \log2$
  \item[or]
  \item\label{i:interesting}
    $\log |\Esf_{\Hsf}|-\frac12\log |\Xsf_{\Hsf}|-\frac12\log| \Ysf_{\Hsf}|
    \leq
    \log\lambda_{2}+\log 2$
  \end{enumerate}
  \unskipp
\end{corollary}
\begin{proof}
  Let $\bar\Hsf$ be the induced closure of $\Hsf$, that is an induced
  subgraph of $\Gsf$ with the same vertex set as $\Hsf$.
  Then we have
  \begin{equation}\label{eq:induced-sub}
    \begin{cases}
      |\Xsf_{\bar\Hsf}|=|\Xsf_{\Hsf}|\\
      |\Ysf_{\bar\Hsf}|=|\Ysf_{\Hsf}|\\
      |\Esf_{\bar\Hsf}|\geq|\Esf_{\Hsf}|
    \end{cases}
  \end{equation}
  Applying Theorem~\ref{p:mixing} to $\bar\Hsf$ we obtain
  \begin{equation}\label{eq:mixing}
    |\Esf_{\bar\Hsf}|\leq s + t
  \end{equation}
  where
  \begin{align*}
    s
    &:=
      |\Esf|\cdot\frac{|\Xsf_{\bar\Hsf}|}{|\Xsf|}
      \cdot\frac{|\Ysf_{\bar\Hsf}|}{|\Ysf|}\\
    t
    &:=
      \lambda_{2}\cdot |\Xsf_{\bar\Hsf}|^{1/2}
      \cdot|\Ysf_{\bar\Hsf}|^{1/2}
  \end{align*}
  Note that the logarithm of the sum of two numbers is equal to the
  logarithm of the maximum with error not exceeding $\log2$.  Thus we
  consider now two cases.
  \paragraph{Case 1: $\bm {t\leq s}$.} Then~\eqref{eq:mixing}
  implies
  \[
    \log|\Esf_{\bar\Hsf}|\leq \log2s
  \]
  After substituting $s$ and rearranging the terms we obtain
  \[
    \log |\Esf_{\bar\Hsf}|-\log |\Xsf_{\bar\Hsf}|-\log |\Ysf_{\bar\Hsf}|
    \leq \log |\Esf| - \log |\Xsf| - \log |\Ysf| + \log2
  \]
  After substitution~\eqref{eq:induced-sub} we get the
  alternative~\ref{i:boring} of the corollary.
  \paragraph{Case 2: $\bm {t>s}$.} In this case inequality~\eqref{eq:mixing}
  implies
  \[
    \log|\Esf_{\bar\Hsf}|\leq \log2t
  \]
  Likewise, we substitute the expression for $t$ and rearrange the summands
  to obtain
  \[
    \log |\Esf_{\bar\Hsf}|-\frac12\log |\Xsf_{\bar\Hsf}|-\frac12\log| \Ysf_{\bar\Hsf}|
    \leq
    \log\lambda_{2}+\log 2
  \]
  This, in turn, implies alternative~\ref{i:interesting} of the
  corollary in view of \eqref{eq:induced-sub}.
\end{proof}

If the density of the graph is bounded away from 1, the second largest
eigenvalue cannot be arbitrarily small.
The precise statements is
Proposition~\ref{p:alonlike} below. 
Recall that $\lambda_{1}(\Gsf)=\sqrt{d_{1}(\Gsf) \cdot d_{2}(\Gsf)}$ and its multiplicity
  is equal to one if and only if $\Gsf$ is connected.
\begin{proposition}\label{p:alonlike}
  For every $\epsilon>0$, there is $C>0$ such that for any biregular
  bipartite graph $\Gsf:=(\Xsf\sqcup\Ysf,\Esf)$ with
  \(
  |\Esf|\leq \exp(-\epsilon) \cdot  |\Xsf|\cdot|\Ysf|
  \)
  holds
  \begin{align*}
    \tag{i}\label{eq:alonlike-strong}
    2\log\lambda_{2}(\Gsf)
    &\geq
      \log\max\set{d_{1},d_{2}}-C,\\
    \tag{ii}\label{eq:alonlike-weak}
    2\log\lambda_{2}(G)
    &\geq\log\lambda_{1}(G) -C .
  \end{align*}
  \unskipp[-4ex]
\end{proposition}

The second inequality in the proposition above follows immediately
from the first and is very similar in spirit to the Alon--Boppana or
A.~Nilli bounds, \cite{alon1986eigenvalues,nilli1991second}.  The
statement of Proposition~\ref{p:alonlike}\eqref{eq:alonlike-strong} in
the required form follows from
\cite[corollary~4]{hoholdt2012eigenvalues}.

\subsection{Graphs and random variables}\label{s:graphs-rvs}
Let $\Gbf:=(X,Y)$ be a pair of uniform on the support random
variables. We associate with $\Gbf$ a bi-regular bipartite graph
$\Gsf:=(\Xsf\sqcup\Ysf, \Esf)$, where the left part is the support
$\Xsf$ of $X$, the right part is the support $\Ysf$ of $Y$, and
$\Esf:=\set{(x,y)\st p(x,y)>0}$. We say that $\Gbf$ is
\emph{supported} on $\Gsf$, or $\Gsf$ is the \emph{support} of $\Gbf$,
and write
\[
  \Gsf=\supp\Gbf=\supp(X,Y)
\]

Since the pair $\Gbf$ is uniform on its support, we have
\begin{equation}
  \label{eq:graph-rv}
  \begin{aligned}
  &H(X)=\log|\Xsf|,
  && H(XY)=\log|\Esf|,\\
  &H(Y)=\log|\Ysf|,
  &&\L(X,Y)=\log\lambda_{1}(\Gsf).
\end{aligned}
\end{equation}

Recall that $\lambda_{2}(\Gsf)$ stands for the second largest
eigenvalue of the graph. The next Corollary follows from
Proposition~\ref{p:alonlike}\eqref{eq:alonlike-weak} by applying it to
the biregular bipartite graph supporting $\Gbf$ and using
the identities~\eqref{eq:graph-rv}.
\begin{corollary}\label{p:alonlike-rv} For every $\epsilon>0$, there
  is a constant $C=C(\epsilon)$, such that for any pair $\Gbf=(X,Y)$
  uniformly supported on $\Gsf$ with
  $I(X; Y)\geq\epsilon$ holds
  \[
    2\log\lambda_{2}(\Gsf)
    \geq
    \L(X,Y) - C.
  \]
  \unskipp
\end{corollary}

\section{Ingleton bounds for almost-uniform tuples}\label{s:ing-q-unif}
In this section we consider a pair $\Gbf=(X,Y)$ uniform on its support
and derive lower bound on $\ing(X,Y,A,B)$ under suitable assumptions
on the extending variables $(A,B)$. We will show that $\ing(X,Y,A,B)$
is bounded from below in terms of the spectrum of the graph $\Gsf$ supporting $\Gbf$.
  
Note that for all $(X,Y,A,B)$
\[
  \ing(X,Y,A,B)\geq -I(X; Y)
\]
Therefore, if $I(X; Y)$ is small, then it is trivial that
$\ing(X,Y,A,B)$ is already ``not too negative''. We therefore fix some $\epsilon_{0}>0$ (say, $\epsilon_{0}=1$)
and \textbf{assume throughout the article}, that
\begin{equation}
  \tag{$*$}\label{eq:I(X; Y)>C}
  I(X; Y)\geq\epsilon_{0}
\end{equation}
Denote by $C_{0}$ the constant provided by
Proposition~\ref{p:alonlike} or Corollary~\ref{p:alonlike-rv} for the
chosen value $\epsilon=\epsilon_{0}$, so that under the assumption
\eqref{eq:I(X; Y)>C} we have the bound
\begin{equation}
  \tag{$**$}\label{eq:assumption}
  2\log\lambda_2(\Gsf)
  \geq
  \log\lambda_{1}(\Gsf) - C_{0}
  =\L(X,Y)-C_{0},
\end{equation}
where $L(X,Y)$ is computed for the uniform random pair $(X,Y)$ supported on~$\Gsf$.

Recall that throughout the paper the $O(\cdot)$-notation is understood so that
$0\leq O(x)\leq C_{1}\cdot x + C_{2}$,
where $C_{1}, C_{2}\ge 0$ are constants depending on the context.
Thus, the $O(\cdot)$-terms may depend implicitly on the chosen threshold $\epsilon_{0}$ (via
$C_{0}$) but on no other implicit parameters. Later,
in Section~\ref{s:explict-constants}, we briefly discuss the
possibility of obtaining bounds with explicit constants and outline
the main obstacles to this goal.

We first assume that the entire quadruple $(X,Y,A,B)$ is
uniform on its support, and then extend the argument to a more general
setting.  The lemma and proposition that apply only to uniform
distributions (proven in Section~\ref{s:ing-unif}) will not be used
directly in the proof of our main result;  rather, we present this
restricted version of the argument to illustrate the main ideas of the
spectral technique in a simpler setting and to provide a gentle
introduction to the proof of the main theorem.
  
\subsection{Uniform tuples}\label{s:ing-unif}
\begin{lemma}\label{p:unif}
  Let $\Gbf=(X,Y)$ be a pair of random variables uniform on its
  support graph $\Gsf$. Let $(X,Y,A)$ be its uniform on its
  support extension. Then at least one of the two conditions below
  holds true. Either
  \begin{enumerate}[label=(\roman*)]
  \item \label{i:unif-boring}
    $I(X; Y|A)\geq I(X; Y)-\log 2$
  \item[or]
  \item\label{i:unif-interesting}
    $\L(X,Y|A)\leq
    \log\lambda_{2}(\Gsf) + \log2$.
  \end{enumerate}
  \unskipp
\end{lemma}

\begin{proof}
  Let $\Gsf=(\Xsf\sqcup\Ysf,\Esf)$ be a bipartite graph supporting the
  pair $\Gbf$.  For an atom $a\in\Asf$, the support of the distribution
  $(X,Y|a)$ is a subgraph
  \[
    \Hsf_{a}=(\Xsf_{a}\sqcup\Ysf_{a},\Esf_{a})\subset\Gsf.
  \]
  Because of uniformity of the triple, the distribution $(X,Y|a)$
  is also uniform on its support.  Uniformity of the triple also implies that
  $H(X|a)$, $H(Y|a)$ and $H(XY|a)$ are independent of
  $a\in\Asf$ and therefore equal to $H(X|A)$, $H(Y|A)$ and $H(XY|A)$,
  respectively. Thus
  \begin{equation*}
    \begin{split}
      H(X|A)&=H(X|a)=\log|\Xsf_{a}|\\
      H(Y|A)&=H(Y|a)=\log|\Ysf_{a}|\\
      H(XY|A)&=H(XY|a)=\log|\Esf_{a}|\\
    \end{split}
  \end{equation*}
  Apply Corollary~\ref{p:mixing-log} to the subgraph $\Hsf_{a}$ using
  the entropy values above. After rearranging the terms we obtain the
  required inequalities.
\end{proof}

\begin{proposition}\label{p:ing-unif}
  Let the quadruple $(X,Y,A,B)$ be uniform on its support and the pair $(X,Y)$
  satisfy condition~\eqref{eq:I(X; Y)>C} on page~\pageref{eq:I(X; Y)>C}.  Then
  \[
    \ing(X,Y,A,B)
    \geq
    -2\log\lambda_{2}(\Gsf)+\log\lambda_{1}(\Gsf)-O(1)
  \]
  where $\Gsf:=\supp (X,Y)$.
\end{proposition}

\begin{proof}
  To prove the proposition consider the triples $(X,Y,A)$ and
  $(X,Y,B)$ and apply Lemma~\ref{p:unif} to each.  Consider two cases.
  \paragraph{Case 1: At least one of the following two inequalities
    hold.}
  \begin{equation*}
    \left[
      \begin{aligned}
        I(X; Y|A) &\geq I(X; Y)-\log2\quad\text{or}\\
        I(X; Y|B) &\geq I(X; Y)-\log2
      \end{aligned}
    \right.
  \end{equation*}
  Substituting one of these two conditions in the expression for
  $\ing(X,Y,A,B)$ and dropping the non-negative terms we get
  \[
    \ing(X,Y,A,B)\geq -\log2.
  \]
  Together with the inequality~\eqref{eq:assumption} on
  page~\pageref{eq:assumption} we obtain the conclusion of the
  proposition in this case.
  \paragraph{Case 2: Both of the following two inequalities hold.}
  \begin{equation*}
    \left\{
      \begin{aligned}
        \L(X,Y|A)
        &\leq
          \log\lambda_{2}(\Gsf) + \log2
          \quad\text{and}\\
        \L(X,Y|B)
        &\leq
          \log\lambda_{2}(\Gsf) + \log2.
      \end{aligned}
    \right.
  \end{equation*}
  These inequalities together with inequality~\ref{p:shannon-ineq}\ref{i:ing-lll} give
  \begin{align*}
    \ing(X,Y,A,B)
    &\geq
      -\L(X,Y|A)-\L(X,Y|B)+\L(X,Y)\\
    &\geq -2\log\lambda_{2}(\Gsf)+ \log\lambda_{1}(\Gsf) - \log4 .
  \end{align*}
  \unskipp
\end{proof}

\subsection{Almost-uniform tuples}\label{ss:ing-q-unif}
The argument presented in the previous section extends, with only
minimal modifications, to a more general setting.  First, the
uniformity assumption can be relaxed to $\delta$-\emph{uniformity}.
Second, instead of the second largest eigenvalue of the graph $\Gsf$
supporting $\Gbf$, we may take the second eigenvalue of a bigger
graph, which contains $\Gsf$.  To make this observation precise, we
introduce the following definition.
\begin{definition}
  Let $\Gsf=(\Xsf\sqcup\Ysf,\Esf)$ be a biregular bipartite graph.
  We say that a pair of random variables $\Gbf'=(X',Y')$ with joint
  distribution $p$ is \emph{subsupported} on $\Gsf$, written
  \[
    \supp\Gbf'\subset\Gsf,
  \]
  if the supports of $X'$ and $Y'$ are subsets of $\; \Xsf$ and
  $\,\Ysf$, respectively, and $p(x',y')>0$ implies $(x',y')\in\Esf$.
  The marginal distributions of $X'$ and $Y'$ need not be fully
  supported on $\Xsf$ and $\Ysf$, respectively, and the support of the
  joint $X'Y'$ need not be equal to $\Esf$.
\end{definition}
Let $\Gbf'=(X',Y')$ be a $\delta$-uniform pair of random variables
subsupported on a biregular bipartite graph
$\Gsf=(\Xsf\sqcup\Ysf,\Esf)$. We show below that the proofs of
Lemma~\ref{p:unif} and Proposition~\ref{p:ing-unif} remain valid, with
minor adaptations, if the second largest eigenvalue of the support graph is replaced by
the second largest eigenvalue of a bigger graph.

\begin{lemma}\label{p:q-unif}
  Let $\Gbf'=(X',Y')$ be a $\delta$-uniform pair of random
  variables subsupported on a biregular bipartite graph
    $\Gsf=(\Xsf\sqcup\Ysf,\Esf)$.  Let $(X',Y',A')$ be a
  $\delta$-uniform extension of $\Gbf'$. Then at least one of the
  two conditions below holds:
  \begin{enumerate}[label=(\roman*)]
  \item \label{i:q-unif-boring}
    $\displaystyle I(X'; Y'|A') \geq \log\frac{|\Xsf| \cdot |\Ysf|}{|\Esf|}-(\log2+ 5 \log\delta)$
  \item[or]
  \item\label{i:q-unif-interesting}
    $\L(X',Y'|A')
    \leq
    \log\lambda_{2}(\Gsf) + \log2 + 4\log\delta$
  \end{enumerate}
\end{lemma}

\begin{proof}
  As usual, let
  \[
    \Xsf':=\supp X'\subset\Xsf,\quad
    \Ysf':=\supp Y'\subset\Ysf,\quad
    \Esf':=\supp X'Y'\subset\Esf
  \]
  Further, for each atom $a\in\Asf'$ we
  denote by $\Xsf_{a}$, $\Ysf_{a}$, and $\Esf_{a}$ the supports of the
  distributions of $X'|a$, $Y'|a$ and $X'Y'|a$, respectively, and by
  $\Hsf_{a}:=(\Xsf_{a}\sqcup\Ysf_{a},\Esf_{a})$ the subgraph of $\Gsf$.
  Since the triple $(X',Y',A')$ is $\delta$-uniform, the cardinalities of
  $\Xsf_{a}$, $\Ysf_{a}$, and $\Esf_{a}$ for different choices of $a$
  differ at most by a factor of $\delta$.  Define
  \begin{equation}
      I_{\max}
      :=
      \max\set{\log|\Xsf_{a}|+\log|\Ysf_{a}|-\log|\Esf_{a}|\st
        a\in \Asf'}.      
  \end{equation}
  For any atom $a\in\Asf'$ we have the bounds
  \[
    I_{\max}\geq\log|\Xsf_{a}|+|\log\Ysf_{a}|-\log|\Esf_{a}|\geq I_{\max}-3\log\delta.
  \]
  By Lemma~\ref{p:entropy-d-uniform} we then have
  \begin{align*}
    I(X'; Y'|A')
    &=
      \sum_{a\in\Asf'}p(a)I(X'; Y'|a) \\
    &=
      \sum_{a\in\Asf'}p(a) \big( H(X'|a) + H(Y'|a) - H(X'Y'|a) \big) \\
    &\geq
      \sum_{a\in\Asf}p(a)
      \left(
      \log|\Xsf_{a}| - \log\delta +|\log\Ysf_{a}| - \log\delta
      -\log|\Esf_{a}|
      \right) \\
    &\geq
      I_{\max}-5\log\delta.
  \end{align*}
  Suppose alternative~\ref{i:q-unif-boring} does not hold, that is
  \[
    I(X'; Y'|A')
    <
    \log\frac{|\Xsf| \cdot |\Ysf|}{|\Esf|}-(\log2+5\log\delta).
  \]
  Hence, for every $a\in\Asf'$
  \begin{align*}
    \log|\Xsf_{a}|+|\log\Ysf_{a}|-\log|\Esf_{a}|
    &\leq I_{\max}\leq I(X'; Y'|A') + 5\log\delta \\
    &<\log\frac{|\Xsf| \cdot |\Ysf|}{|\Esf|} -\log2\\
    &\le\log|\Xsf|+|\log\Ysf|-\log|\Esf| -\log2. 
  \end{align*}
  Therefore, for any atom $a\in\Asf'$ the subgraph $\Hsf_{a}$ satisfies
  alternative~\ref{i:interesting} of Corollary~\ref{p:mixing-log}
  applied to the graph $\Gsf$.
  That is, 
  \[
    \log |\Esf_a|-\frac12\log |\Xsf_a|-\frac12\log| \Ysf_a|
    \leq
    \log\lambda_{2}(\Gsf)+\log 2.
  \]
  On the other hand, for every $a\in\Asf'$ we have
  \begin{equation*}
    \begin{aligned}
      \big|\log |\Esf_a|-H(X'Y'|A')\big|&\leq 2\log\delta,\\
      \big|\log |\Xsf_a|-H(X'|A')\big|&\leq 2\log\delta,\\
      \big|\log |\Ysf_a|-H(Y'|A')\big|&\leq 2\log\delta.
    \end{aligned}
  \end{equation*}
  Therefore, 
  \begin{equation*}
    \begin{aligned}
      \L(X',Y'|A') 
      &=H(X'Y'|A')-\frac12H(X'|A')-\frac12H(Y'|A')\\
      &\leq
        \log\lambda_{2}(\Gsf)+\log 2+4\log\delta.
    \end{aligned}
  \end{equation*}
  Thus, we get the conclusion~\ref{i:q-unif-interesting} of
  Lemma~\ref{p:q-unif}.
\end{proof}

\begin{proposition}\label{p:ing-q-unif}
  Let $\Gbf'=(X',Y')$ be a $\delta$-uniform pair of random variables,
  and suppose that $\Gbf'$ is subsupported on a biregular bipartite
  graph $\Gsf=(\Xsf\sqcup\Ysf,\Esf)$ that satisfies
  condition~\eqref{eq:assumption} on page~\pageref{eq:assumption}.
  Assume that the quadruple $(X',Y',A',B')$ is a $\delta$-uniform
  extension of $\Gbf'$. Then
  \[
    \ing(X',Y',A',B')\geq
    -2\log\lambda_{2}(\Gsf) + \L(X',Y')-O(\log\delta).
  \]
  \unskipp[-7ex]
\end{proposition}
Note the summand $\L(X',Y')$ in the proposition above, which,
in general, is different from $\log\lambda_{1}(\Gsf)$.
\begin{proof}
  Similarly to the previous proof, we set
  \[
    \Xsf':=\supp X'\subseteq\Xsf,\quad
    \Ysf':=\supp Y'\subseteq\Ysf,\quad
    \Esf':=\supp X'Y'\subseteq\Esf.
  \]
  Then we have
  \begin{equation}\label{eq:32b-estimate}
    H(X')\leq\log|\Xsf|,\quad
    H(Y')\leq\log|\Ysf|.
  \end{equation}

  We proceed along the lines of the the proof of
  Proposition~\ref{p:ing-unif} by first applying Lemma~\ref{p:q-unif}
  to each triple $(X',Y',A')$ and $(X',Y',B')$ and distinguishing two cases
  \paragraph{Case 1: At least one of the following two inequalities hold.}
  \begin{equation*}
    \left[
      \begin{aligned}
        I(X'; Y'|A')
        &\geq
          \log\frac{|\Xsf| \cdot |\Ysf|}{|\Esf|} -
          (\log2+5\log\delta)
          \quad\text{or}\\
        I(X'; Y'|B')
        &\geq
          \log\frac{|\Xsf| \cdot |\Ysf|}{|\Esf|}
          -(\log2+ 5\log\delta).
      \end{aligned}
    \right.
  \end{equation*}
  Without loss of generality assume that the first of these two
  inequalities holds.
  We substitute it in the expression for $\ing(X',Y',A',B')$ and obtain
  \[
    \begin{aligned}
      \ing(X',Y',A',B')
      & = I(X'; Y'|A') + I(X'; Y'|B') + I(A'; B') - I(X'; Y') \\
      &\geq I(X'; Y'|A')  - I(X'; Y') \\
      &\geq
        \log\frac{|\Xsf| \cdot |\Ysf|}{|\Esf|} - I(X'; Y') - O(\log\delta)\\ 
      &\geq
        \log\frac{\sqrt{|\Xsf|\cdot|\Ysf|}}{|\Esf|} +
        \log \sqrt{|\Xsf| \cdot |\Ysf|}  - I(X'; Y') - O(\log\delta)
      \\
      &=-\log\lambda_{1}(\Gsf) +
        \frac12 \log|\Xsf| +\frac
        12\log |\Ysf| - I(X'; Y') - O(\log\delta)
      \\
      \bm(\parbox[c]{20mm}{\raggedright By
      Eq. \eqref{eq:assumption}}\bm)\mkern28mu
      &\geq -2\log\lambda_{2}(\Gsf) + \frac12 \log|\Xsf| +\frac
        12\log |\Ysf|
        - I(X'; Y') - O(\log\delta)\\
      \bm(\parbox[c]{20mm}{\raggedright By
      Eq. \eqref{eq:32b-estimate}}\bm)\mkern28mu      
      &\geq
        -2\log\lambda_{2}(\Gsf)  + \frac12 H(X')+\frac12 H(Y') -
        I(X'; Y') - O(\log\delta)  \\    
      &=
        -2\log\lambda_{2}(\Gsf)  +  \L(X',Y') - O(\log\delta).    
    \end{aligned}
  \]
  
  \paragraph{Case 2: Both of the following two inequalities hold.}
  \begin{equation*}
    \left\{
      \begin{aligned}
        \L(X',Y'|A')
        &\leq
          \log\lambda_{2}(\Gsf) + \log2 + 4\log\delta
          \quad\text{and}\\
        \L(X',Y'|B')
        &\leq
          \log\lambda_{2}(\Gsf) + \log2 +4\log\delta.
      \end{aligned}
    \right.
  \end{equation*}
  Then we proceed with the bounds using these inequalities and
  inequality~\ref{p:shannon-ineq}\ref{i:ing-lll}:
  \begin{align*}
    \ing(X',Y',A',B')
    &\geq
      -\L(X',Y'|A') - \L(X',Y'|B') + \L(X',Y')\\
    &\geq -2\log\lambda_{2}(\Gsf) +
      \L(X',Y') - O(\log\delta).
  \end{align*}
  \unskipp[-7ex]
\end{proof}

One natural situation in which a pair of random variables becomes
subsupported on a larger graph arises when a uniform pair is
conditioned on the value of a third random variable jointly
distributed with the pair.  This case is addressed in the following
corollary.
\begin{corollary}\label{p:ing-q-unif-conditionalv2}
  Suppose that the pair $(X,Y)$ is uniform
   and satisfies condition~\eqref{eq:I(X; Y)>C} on page~\pageref{eq:I(X; Y)>C},
  and the extension
  $(X,Y,A,B,V)$ is such that for some fixed atom $v_{0}\in\Vsf$
  the quadruple $(X,Y,A,B|v_{0})$ is $\delta$-uniform.
  Let $\Gsf:=\supp(X,Y)$.
  Then
  \[
    \ing(X,Y,A,B|v_{0})
    \geq
    -2\log\lambda_{2}(\Gsf)+ \L(X,Y|v_{0}) -O(\log\delta)      
  \]
  \unskipp
\end{corollary}
\begin{proof}
  For the atom $v_{0}$ from the support of $V$, the pair $(X,Y|v_{0})$
  is $\delta$-uniform and subsupported on $\Gsf$.  Since the pair
  $(X,Y)$ satisfies assumption~\eqref{eq:I(X; Y)>C}, its support
  satisfies condition~\eqref{eq:assumption} on
  page~\pageref{eq:assumption}.  Therefore, we may apply
  Proposition~\ref{p:ing-q-unif} to the distribution $(X,Y,A,B|v_{0})$, 
  which gives the conclusion of the corollary.
\end{proof}

\section{How to decompose a tuple into almost-uniform tuples}\label{s:split}
In this section we will show how to decompose (a tuple of) random
variables as a mixture of almost-uniform ones and a small ``remainder''.
\subsection{Splitting a single random variable}\label{s:split1}
\begin{proposition}\label{p:split1}
  For every random variable $X$ there exists an extension $(X,U)$ such
  that the alphabet of $U$ is
  $\Usf=\set{u_{0},u_{1},\dots,u_{k-1},u_{\infty}}$ and the following
  properties hold:
  \begin{enumerate}[label=(\roman*)]
  \item \label{i:partition}
    $H(U|X)=0$; 
  \item \label{i:smalltail}
    $\displaystyle\Pbb[U=u_{\infty}]\cdot H(X)\leq\frac{1}{\log2}$; 
  \item \label{i:entropybound}
    $k\leq H(X)^{2}+1$; 
  \item \label{i:uniform}
    $X|u_i$ is 2-uniform for every $i\in[k]$; 
  \item \label{i:smallparts}
    $\sharp(X|u_{i})\leq 2^{H(X)^{2}+1}$ for every $i\in[k]$.
  \end{enumerate}
  \unskipp
\end{proposition}
\begin{proof}
  Let $p$ be the distribution mass function of $X$.
  We construct $U$ by partitioning the alphabet
  $\Xsf$ of $X$, where the alphabet of $U$ is  the set of
  parts of the partitioning;  $U$ will be deterministic function of $X$
  and property~\ref{i:partition} is satisfied. We will
  partition $\Xsf$ as
  \begin{equation*}
    \Xsf=u_{0}\sqcup u_{1}\sqcup\dots\sqcup u_{k-1}\sqcup u_{\infty}
  \end{equation*}
  where
  \begin{equation*}
    \begin{aligned}
      u_{i}
      &:=
        \set{x\in\Xsf\st 2^{-(i+1)}<p(x)\leq 2^{-i}},
      &&
         i=0,\dots,k-1\\
      u_{\infty}
      &:=
        \set{x\in\Xsf\st p(x)\leq 2^{-k}}
    \end{aligned}
  \end{equation*}
  and value of $k$ is chosen later in the proof.
  If some of the parts defined above are empty, we silently drop  them.
  We set
  \[
    q_{i}:=\Pbb[U=u_{i}]=\sum_{x\in u_{i}} p(x)>0,
    \quad
    i=0,\ldots,k-1,\infty.
  \]
  Note that for $x,y\in u_{i}$, $i\in[k]$, the probabilities $p(x)$
  and $p(y)$ differ at most 2-fold, so that $X|u_{i}$ is 2-uniform
  and assertion~\ref{i:uniform} is satisfied.

  It remains to choose value $k$ so that
  properties~\ref{i:smalltail},~\ref{i:entropybound}
  and~\ref{i:smallparts} of the proposition hold.  To find an
  appropriate value of $k$, we write
  \begin{equation*}
    H(X)
    \geq
    \sum_{x\in u_{\infty}}p(x)\log p(x)^{-1}
    \geq
    q_{\infty}\cdot\log2^{k}
    =k\cdot q_{\infty}\cdot\log2.
  \end{equation*}
  We take $k:=\left\lceil H(X)^{2}\right\rceil$.  Observe that
  $k \le H(X)^2+1$, so property~\ref{i:entropybound} follows
  immediately.  Further, for the chosen $k$
   \[
    q_{\infty}\cdot H(X)\leq\frac{1}{\log2}
  \]
  and property~\ref{i:smalltail} is satisfied.
  We also observe that
  $|u_{i}|\leq q_{i}\cdot2^{i+1}$, thus
  \begin{equation*}
    \sharp(X|u_{i})\leq q_{i}\cdot 2^{i+1}\leq 2^{k}\leq2^{H(X)^{2}+1}, 
  \end{equation*}
  and property~\ref{i:smallparts} and proposition as a whole are
  proven.
\end{proof}

\subsection{Splitting a tuple into almost-uniform
  tuples}\label{s:split4}
The next claim generalizes Proposition~\ref{p:split1} of the
previous section. It shows that a tuple of random variables can be
decomposed into a mixture of almost-uniform tuples --- at most
$O\big(\log(\text{entropy of the tuple})\big)$ of them --- and a small
residual measure.
    
\begin{theorem}[Almost Uniform Decomposition Lemma]\label{p:split4}
  For every $n\in\Nbb$ there are constants $C_{n}>0$ and
  $\delta_{n}\geq1$ such that for any tuple of random variables
  $\Xbf=(X_{i}\st i\in[n])$ there exists an extension of $\Xbf$ by a
  random variable $V$ with alphabet
  $\Vsf=\set{v_{0},v_{1},\dots,v_{k-1},v_{\infty}}$, such that the
  following properties hold:
  \begin{enumerate}[label=(\roman*)]
  \item\label{i:smalltail4}
    $\displaystyle\Pbb[V=v_{\infty}]\cdot H(X_{[n]})\leq\frac{1}{\log2}$
  \item\label{i:entropybound4}
    $H(V)\leq C_{n}\log(H(X_{[n]})+1\big)$
  \item\label{i:regular4}
    $\Xbf|v_{i}$ is $\delta_{n}$-uniform for every $i\in[k]$.
  \end{enumerate}
  \unskipp[-7ex]
\end{theorem}
Before proving the theorem above we introduce some necessary tools.
We will use a theorem by N.~Alon, I.~Newman, A.~Shen, G.~Tardos,
N.~Vereshchagin, \cite[Theorem 3]{alon2007partitioning}.  To formulate
this result, Theorem~\ref{p:cool} below, additional notation is
needed. Suppose we have a collection of finite sets
$(\Xsf_{i}\st i=0,\dots,n-1)$ and a subset $S\subset\prod
\Xsf_{i}$. For a subsets of indices $I\in 2^{[n]}$ we denote by
$S_{I}$ the image of the projection of $S$ to $\Xsf_{I}$. For
$\sbf\in S_{I}$ we denote by $S|\sbf$ the fiber of this projection
over point $\sbf$. More generally, for a pair of disjoint subsets of
indices $I,J\in2^{[n]}$ and a point $\sbf\in S_{J}$, we set
$S_{I}|\sbf:=S_{I\cup J}|\sbf$. We say that a multidimensional set $S$
is \emph{$\beta$-regular}%
\footnote{ In the original article such sets are called \emph{strongly
    $\beta$-uniform}. This terminology conflicts with ours, so we use
  a different term.  The term ``$\beta$-regular'' seems natural, since
  $\beta$-regular 2D sets correspond to ``almost biregular'' bipartite
  graphs in the sense that, within each part, the ratio between the
  maximum and minimum degrees is universally bounded.} %
for some $\beta\geq1$ if for all disjoint $I,J\in2^{[n]}$ holds
\[
  \frac{\max\set{\sharp(S_{I}|\sbf)\st \sbf\in S_{J}}}
  {\min\set{\sharp(S_{I}|\sbf)\st \sbf\in S_{J}}}
  \leq \beta, 
\]
where $\sharp(\cdot)$ stands for cardinality of the set.

Now we are ready to state the theorem.
\begin{theorem}[{\cite[Theorem 3]{alon2007partitioning}}]\label{p:cool}
  For every $n\in\Nbb$ there exist constants $\alpha_{n}>0$ and
  $\beta_{n}\geq1$ such that for every $n$-dimensional finite set $S$
  there exists a partition of $S$ into at most
  $(\log|S|)^{\alpha_{n}}$ $\beta_{n}$-regular parts.
\end{theorem}

\begin{proof}[Proof of Theorem~\ref{p:split4}]
  We start by applying Proposition~\ref{p:split1} to the random
  variable $X_{[n]}$ and construct a partition 
  \[
  \Usf=\set{u_{0},u_{1},\dots,u_{k'-1},u_{\infty}}
  \]
  of the alphabet $\Xsf_{[n]}$ of $X_{[n]}$.  By
  Proposition~\ref{p:split1}, this partition satisfies properties
  \ref{i:partition}--\ref{i:smallparts}.  In particular, by
  Proposition~\ref{p:split1}\ref{i:entropybound}, the cardinality $k'$
  of this partition is bounded by $H(X)^2+1$.  (Here $k'$ is used to
  distinguish this quantity from the parameter $k$ appearing in the
  statement of the theorem.)

  We now refine this partition
  using Theorem~\ref{p:cool}.
  We leave $u_{\infty}$ unchanged and set
  $v_{\infty}:=u_{\infty}$, thus guaranteeing
  property~\ref{i:smalltail4} in the conclusion of the theorem.

  Consider one of the parts $u_{i}$, $i\in[k']$.  By
  Proposition~\ref{p:split1}\ref{i:uniform} and \ref{i:smallparts},
  $X_{[n]}|u_{i}$ is 2-uniform and its support has at most
  $2^{H(X)^2+1}$ points. Apply Theorem~\ref{p:cool} to
  \[
    S:=u_{j}\subset\Xsf_{[n]}=\prod_{i\in[n]}\Xsf_{i}.
  \]
  The theorem provides the partition of $u_{i}$ into at most
  \begin{equation}
  \label{eq:upper-bound-parts}
    \left(\log2^{H(X_{[n]})^{2}+1}\right)^{\alpha_{n}}
    =
    (\log2)^{\alpha_{n}}\big(H(X_{[n]})^{2}+1\big)^{\alpha_{n}}
  \end{equation}
  $\beta_{n}$-regular parts $v_{i1},v_{i2},\dots$ Thus, $\Xbf|v_{ij}$
  is a $\beta_{n}$-regular $n$-tuple for every $j$ in its range. Note
  also that $X_{[n]}|v_{ij}$ is 2-uniform, therefore by
  Lemma~\ref{p:unif-reg}\ref{i:reg->unif}, the $n$-tuple
  $\Xbf|v_{ij}$ is $2\beta_{n}$-uniform. Thus, we can take
  $\delta_{n}=2\beta_{n}$, proving property~\ref{i:regular4}.
  Refining all parts $u_i$, $i\in[k']$ gives at most
  \[
    D:=(\log2)^{\alpha_{n}}\big(H(X_{[n]})^{2}+1\big)^{\alpha_{n}+1}
  \]
  parts (the number $D$ is the product of the upper bound on
    $k'$ and the upper bound on the number of almost-regular parts
    given in \eqref{eq:upper-bound-parts}). Therefore,
  \[
    H(V)
    \leq
    \log (D+1)
    \leq
    C_{n}\cdot\log (H(X_{[n]})+2)
  \]
  for some $C_{n}>0$ (the added constant $+2$ in the logarithmic term
  handles the degenerate case $H(X_{[n]})=0$).  This proves the
  required inequality in~\ref{i:entropybound4}.
\end{proof}

\section{General case}\label{s:main}
\def\thetheorem{\thesection.\Alph{theorem}}
\begin{theorem}\label{p:main}
  Suppose $\Gbf=(X,Y)$ is a pair of random variables uniform on its support
  satisfying condition~\eqref{eq:I(X; Y)>C}. Let $\Gsf=\supp \Gbf$.
  Then for \emph{any} extension $(X,Y,A,B)$
  the following bound holds: 
  \begin{equation*}
    \ing(X,Y,A,B)
    \geq
    -2\log\lambda_{2}(\Gsf) + \log\lambda_{1}(\Gsf) - O\big(\log H(\XYAB)\big). 
  \end{equation*}
  \unskipp[-4ex]
\end{theorem}
\begin{remark}
  There might be a temptation to remove the sublinear term
  $O\big(\log H(\XYAB)\big)$ by the standard tensorization argument:
  take $n$ i.i.d. copies of $(X,Y,A,B)$, apply the inequality to
  the ``tensorized'' distribution, normalize by $n$ and send
  $n$ to infinity. However, the inequality depends on the spectral
  properties of the graph corresponding to the uniform pair $(X,Y)$.
  Under tensorization, the second largest eigenvalue will converge
  (on normalized $\log$-scale) to the largest eigenvalue, thus
  rendering the inequality trivial and useless.
\end{remark}

\begin{proof}   
  Let $V$ be the extension of $(X,Y,A,B)$ provided by 
  Theorem~\ref{p:split4}.  Then
  \begin{align}\label{eq:ing}
    \ing(X,Y,A,B)
    &\geq
      \ing(X,Y,A,B|V) - O(H(V)),
  \end{align}
  since each of four summands in the expression for $\ing$ may shift
  under conditioning by at most $H(V)$.%
  \footnote{It is not hard to show that
    $\ing(X,Y,A,B)\geq \ing(X,Y,A,B | V) - 4H(V)$.  An anonymous
    referee observed that the even stronger inequality
    $\ing(X,Y,A,B)\geq \ing(X,Y,A,B | V) - 3H(V)$ is itself a
    Shannon-type inequality. Since the optimal constant in front of
    the term $H(V)$ is immaterial for our purposes, 
    we omit the proof from the present version of the paper.
     } %
  Let $\Vsf=\set{v_{0},\dots,v_{k-1},v_{\infty}}$ be the alphabet of
  $V$ as in Theorem~\ref{p:split4} and $q_{i}:=\Pbb[V=v_{i}]$,
  $i=0,\dots,k-1,\infty$.  Consider the first summand in
  \eqref{eq:ing} and decompose
  \begin{align}\label{eq:ing|V}
    \ing(X,Y,A,B|V)
    &=
      \sum_{i=0}^{k-1}q_{i}\ing(X,Y,A,B|v_{i})
    + q_{\infty} \ing(X,Y,A,B|v_{\infty}). 
  \end{align}
  We first bound the residual term. Since $(X,Y)$ is uniform on its
  support, conditioning on an atom of $V$ does not increase the entropy.
  Using the following crude estimate
  \[
    \ing(X,Y,A,B|v_{\infty})\geq -I(X; Y|v_{\infty})\geq -H(XY)\geq-H(\XYAB)
  \]
  we get
  \begin{equation}\label{eq:tailbound}
    q_{\infty} \ing(X,Y,A,B|v_{\infty})\geq -q_{\infty}H(\XYAB)\geq -\frac{1}{\log2}, 
  \end{equation}
  where the last inequality follows from
  Proposition~\ref{p:split4}\ref{i:smalltail4}.

  Before we estimate the first summand in~\eqref{eq:ing|V},  we need to
  prepare some auxiliary inequalities:\\
  \textbf{1.} Note that $(X,Y)$ is a uniform pair, thus conditioning does not
  increase entropy and we have
  \begin{equation}
    \label{eq:tailbound-l}
    q_{\infty}\cdot\L(X,Y|v_{\infty})
    \leq
    q_{\infty}\cdot\L(X,Y)
    \leq
    q_{\infty}\cdot H(\XYAB)
    \leq
    \frac{1}{\log2}. 
  \end{equation}
  \textbf{2.} Since $\lambda_{1}\ge 1$, and we work under the standing assumption~\eqref{eq:assumption} throughout this paper, 
   we have
    \begin{equation}
    \label{eq:tailbound-lammbda2}
    \begin{aligned}
      2\cdot q_{\infty}\cdot\log\lambda_{2}
      &\geq
        q_{\infty}(\log\lambda_{1}-C_{0})\geq -C_{0}. 
    \end{aligned}
  \end{equation}

  \medskip
  
  Now we turn to the first summand in~\eqref{eq:ing|V}.
  For each $i\neq\infty$,  the quadruple $(X,Y,A,B|v_{i})$ is
  $\delta_{4}$-uniform, so we can apply
  Corollary~\ref{p:ing-q-unif-conditionalv2} to it and obtain 
  \begin{equation}\label{eq:ing|v}
    \ing(X,Y,A,B|v_{i})
    \geq
    -2\log\lambda_{2}(\Gsf)+\L(X,Y|v_{i})-O(1).  
  \end{equation}
  Therefore
  \begin{equation}\label{eq:mainbound}
  \begin{aligned}
    \sum_{i=0}^{k-1}&q_{i}\cdot\ing(X,Y,A,B|v_{i})
    \geq
      \sum_{i=0}^{k-1}q_{i}
      \big(
      -2\log\lambda_{2}(\Gsf)+\L(X,Y|v_{i}) - O(1)
      \big)\\
    &\geq
      -2\log\lambda_{2}(\Gsf)+\L(X,Y|V)\\
    &\quad\,
      - q_{\infty}\L(X,Y|v_{\infty})+2\cdot
      q_{\infty}\cdot\log\lambda_{2}(\Gsf) - O(1)\\
    &\geq
      -2\log\lambda_{2}(\Gsf)+\L(X,Y|V) - O(1)\\
    &\geq
      -2\log\lambda_{2}(\Gsf)+\L(X,Y)- O\big(H(V)\big), 
  \end{aligned}
  \end{equation}
  where we used inequalities~\eqref{eq:tailbound-l}--\eqref{eq:ing|v}
  and $|\L(X,Y)-L(X,Y|V)|\leq H(V)$.
  Assembling
  together~\eqref{eq:ing},~\eqref{eq:ing|V},~\eqref{eq:tailbound} and~\eqref{eq:mainbound}, 
  we obtain
  \begin{align*}
    \ing(X,Y,A,B)
    &\geq
      -2\log\lambda_{2}(\Gsf)+\L(X,Y)- O\big(H(V)\big)\\
    &=
      -2\log\lambda_{2}(\Gsf)+\log\lambda_{1}(\Gsf) - O\big(\log H(\XYAB)\big), 
  \end{align*}
  as claimed.
\end{proof}

\section{Examples}\label{s:examples}

In this section, we consider several examples of pairs $(X,Y)$
  uniform on their support whose supporting graphs have strong
  expander properties.  For each example, our main theorem provides a
  strong lower bound on the $\ing$-expression, implying that the
  Ingleton inequality holds up to a logarithmic error term for
  arbitrary auxiliary random variables $A,B$.

By contrast, assuming Conjecture~\ref{co:nonextractable}, the standard
approach via extractable mutual information fails  on
these examples: for every auxiliary random variable $W$, at least one
of the conditions in~\eqref{eq:common-information} is violated by a
large amount, and the classical method or the MMRV inequality
imply only (nearly) trivial lower bounds for the $\ing$-expression
instead of the Ingleton inequality.

\subsection{Linear projective flags}\label{s:linear}
\paragraph{Construction of $(X,Y)$.}
Let $\Fbb$ be a finite field of large cardinality $q$. Let
$\Xsf:=\Fbb P^{2}$ and $\Ysf:=\overline{\Fbb P}^{2}$ be the projective
plane and dual projective plane over field $\Fbb$, respectively.
Define
\[
  \Esf:=\set{(x,l)\in \Fbb P^{2}\times\overline{\Fbb P}^{2}\st x\in l}
\]
In other words, $\Esf$ consist of incident point-line pairs --- a
point in the projective plane $x\in\Fbb P^{2}$ and a projective line
$l\subset\Fbb P^{2}$ passing through $x$.  Denote by $\Fsf$ the
biregular bipartite graph, $\Fsf=(\Xsf\sqcup\Ysf,\Esf)$ and by
$\Fbf=(X,Y)$ the uniform pair supported on $\Fsf$. Direct calculation
then shows
\begin{equation*}
  \begin{aligned}
    H(X)&=\log|\Xsf|=\log(q^{2}+q+1)=2\log q +O(1),\\
    H(Y)&=\log|\Ysf|=\log(q^{2}+q+1)=2\log q +O(1),\\
    I(X; Y)&=\log(\frac{q^{2}+q+1}{q+1})=\log q + O(1).
  \end{aligned}
\end{equation*}

\paragraph{Spectral bound for the $\ing$-expression.} 
The spectrum of the graph $\Fsf$ is
\[
  \sigma(\Fsf)=\set{\pm\sqrt{q},\pm(q+1)},
\]
see for example~\cite{hoffman1965line}.  This graph is a strong
expander in the sense that its second eigenvalue is small relative to
the largest one and we have
\begin{equation*}
  \log\frac{\lambda_{1}}{\lambda_{2}^{2}}=\log(1+\tfrac1q)>0.
\end{equation*}

Applying Theorem~\ref{p:main} to $\Fbf$ and substituting the value
above we obtain the following corollary.

\begin{corollary}\label{p:linflags}
  Let $(X,Y,A,B)$ be four random variables, where $(X,Y)$ is the
  uniformly random choice of an incident point-line pair in the
  projective plane over the field $\Fbb_{q}$. Then
  \[
    \ing(X,Y,A,B)\geq-O(\log H(\XYAB)).
  \]
  \unskipp
\end{corollary}

\paragraph{Non-extractable mutual information.} 
On the other hand, for these $(X,Y)$, the standard proof of the Ingleton inequality 
(based on extraction of the mutual information) seems to fail completely. Indeed,
Conjecture~\ref{co:nonextractable} would imply that
for the graph at hand, for any $W$ 
\begin{equation*}
  \begin{aligned}
    I(X; Y|W) + I(W; X|Y) + I(W; Y|X)
    &\geq
      \log q - C\cdot\log\log q\\
    &=
      I(X; Y) - C'\cdot\log\log q.
  \end{aligned}
\end{equation*} 
In this case, regardless of the choice of $W$, the MMRV
inequality~\eqref{eq:mmrv} cannot imply a statement stronger than
\begin{equation*}
\label{eq:ingleton-much-weaker}
  \ing(X,Y,A,B)
  \geq
  -I(X; Y) + C'\cdot\log\log q.
\end{equation*}
This is much weaker than the Ingleton inequality.  In fact, it is
close to the trivial (Shannon-type) inequality
$\ing(X,Y,A,B)\geq -I(X; Y)$.

\subsection{Algebraic projective flags}\label{s:algebraic}
\paragraph{Construction of $(X,Y)$.}
Example in the previous section is symmetric with respect to
transposing $X$ and $Y$, due to projective duality. We can replace
projective lines in the previous example by projective plane curves of
some fixed degree $k$. Not to dive into algebraic geometry over finite
fields and in order to construct uniform (rather then almost-uniform)
pairs $(X,Y)$, we consider a special case of graphs of polynomials.

Let $\Fbb_{q}^{(k)}[t]$ be the space of polynomials of degree at most
$k$. Define the left and the right parts of the biregular bipartite
graph $\Fsf_{k}=(\Xsf\sqcup\Ysf,\Esf)$ by
\[
  \Xsf:=\Fbb_{q}^{2},\qquad
  \Ysf:=\Fbb_{q}^{(k)}[t]
\]
and
\[
  \Esf
  :=
  \set{\bm{(}\vect{x\\y},p\bm{)}\in\Fbb_{q}^{2}\times\Fbb_{q}^{(k)}[t]
    \st y=p(x)}
\]
Let $\Fbf_{k}=(X,Y)$ be the uniform pair
supported on $\Fsf_{k}$. Then
\begin{equation*}
  \begin{aligned}
    H(X)&=\log|\Xsf|=2\log q,\\
    H(Y)&=\log|\Ysf|=(k+1)\log q,\\
    I(X; Y)&=\log q.
  \end{aligned}
\end{equation*}
\paragraph{Spectral bound for the $\ing$-expression.}\
One can show that in this graph 
\begin{equation*}
  \begin{aligned}
    \log\lambda_{1}(\Fbf_{k})&=\frac{k+1}2\log q,\\
    \log\lambda_{2}(\Fbf_{k})&=\frac{k}2\log q,
  \end{aligned}
\end{equation*}
see, for  example,~\cite{caillat2024common-arxiv}.  Thus, in this example we get the
smallest possible second eigenvalue compared with the degrees of the
graph, see Proposition~\ref{p:alonlike}\eqref{eq:alonlike-strong}.
For this graph we have
\begin{equation*}
  \log\frac{\lambda_{1}}{\lambda_{2}^{2}}=-\frac{k-1}{2}\log q.
\end{equation*}

Applying Theorem~\ref{p:main} to the setup at hand gives the
following corollary.
\begin{corollary}\label{p:algflags}
  Let $(X,Y,A,B)$ be four random variables, where $(X,Y)$ is the
  uniformly random choice of a point in the plane $\Fbb_{q}^{2}$ and a
  graph of a polynomial of degree at most $k$ passing through this point. Then
  \[
    \ing(X,Y,A,B)\geq -\frac{k-1}{2}\log q -O(\log H(\XYAB))
  \]
  \unskipp
\end{corollary}
Since $I(X; Y) = \log q$ and $\ing(X,Y,A,B)\geq -I(X; Y)$ is a
Shannon-type inequality, this corollary is non-trivial only for
$k=1,2$.  In Section~\ref{s:discussion} we discuss the possibility to
make this bound stronger (and non-trivial for $k>2$).

\paragraph{Non-extractable mutual information.} 
And again, the conventional approach fail to prove the Ingleton
inequality for these $(X,Y)$.  Conjecture~\ref{co:nonextractable}
would imply for these $(X,Y)$ and any $W$
\begin{equation*}
  \begin{aligned}
    I(X; Y|W) + I(W; X|Y) + I(W; Y|X)
    &\geq
      \log q - C\cdot\log\log q\\
    &=
      I(X; Y)- C\cdot\log\log q.
  \end{aligned}
\end{equation*}
Similarly with the case of linear flags, in this case the
MMRV inequality~\eqref{eq:mmrv} gives at best
\begin{equation*}
  \ing(X,Y,A,B)\geq -I(X; Y)+C\cdot\log\log q, 
\end{equation*}
which, in turn, is ``almost Shannon'' inequality (up to a 
$(C\cdot\log\log q)$-summand).

\subsection{Linear flags in higher dimensions}  
\paragraph{Construction of $(X,Y)$.}
The example in Section~\ref{s:linear} can be generalized in a different
way. As above, we fix a finite field $\Fbb_{q}$ and three integers
$0<k<l<n$ and define the biregular bipartite graph of linear
($k$-in-$l$)--flags in $\Fbb_{q}^{n}$ as follows. The left and right parts are the Grassmannians
of $k$- and $l$-dimensional linear subspaces of $\Fbb_q^n$,  and the
edge-set is the space of flags:
\begin{equation*}
  \begin{aligned}
    &\Xsf
      :=
      \gr(n,k)=\set{U\st U\subset\Fbb_{q}^{n},\; \dim U=k}\\
    &\Ysf
      :=
      \gr(n,l)=\,\set{V\st V\subset\Fbb_{q}^{n},\; \dim V=l}\\
    &\Esf
      :=
      \set{(U,V)\st U\in\gr(n,k),\; V\in\gr(n,l),\; U\subset V}\\
    &\Fsf_{kl}^{(n)}
    :=
      (\Xsf\sqcup\Ysf,\Esf)
  \end{aligned}
\end{equation*}
Let $\Fbf_{kl}^{(n)}=(X,Y)$ be a uniform pair supported on the graph
$\Fsf_{kl}^{(n)}$, that is a uniformly random choice of a linear
subspace $V$ of dimension $l$ in $\Fbb^{n}$ and a $k$-dimensional
subspace $U$ in $V$.  To write the dimensions of this graph recall
that \emph{Gaussian binomial coefficient} (the size of the
Grassmannian) is defined as
\newcommand{\qbin}[2]{\left[\!\!\begin{array}{cc}#1\\#2\end{array}\!\!\right]_{q}}
\[
  |\gr(n,k)| =  \qbin{n}{k}
  :=
  \frac
  {(q^{n}-1)(q^{n}-q)\dots(q^{n}-q^{k-1})}
  {(q^{k}-1)(q^{k}-q)\dots(q^{k}-q^{k-1})}
  =q^{k(n-k)+O(1)}
\]
Now we can determine the parameters of the graph and compute the entropies:
\begin{equation*}
  \begin{aligned}
    &H(X)=\log|\Xsf|=
    k(n-k)\log q + O(1),\\
    &H(Y)=\log|\Ysf|=
    l(n-l)\log q + O(1),\\
    &I(X; Y)=
    k(n-l)\log q + O(1).
  \end{aligned}
\end{equation*}
\paragraph{Spectral bound for the $\ing$-expression.}\
From the vertex degrees of this graph, we immediately obtain its largest eigenvalue.
\begin{equation*}
  \begin{aligned}
    &\log\lambda_{1}=\frac{\big(n-(l-k)\big)(l-k)}{2}\log q + O(1).
  \end{aligned}
\end{equation*}
The calculation of the rest of the spectrum of this graph is somewhat lengthy and
not very elucidating, therefore it is not included;  besides the
spectrum is likely to be well known to the specialists.  The second
largest eigenvalue is
\[
  \log\lambda_{2}=\frac{(l-k)\big(n-(l-k)-1\big)}{2}\log q+O(1).
\]
Without loss of generality, assume that $k\leq n-l$ (otherwise we can
switch to the dual picture).  Then the largest degree is
\[
  \log d_{1}=(l-k)(n-l)\log q + O(1).
\]
Thus, the graph $\Fsf^{(n)}_{kl}$ is an expander for $k=1$ and
arbitrary $l$;  it attains the lower bound
in Proposition~\ref{p:alonlike}\eqref{eq:alonlike-strong} up to an
additive $O(1)$ term.
For this graph we have
\begin{equation*}
  \log\frac{\lambda_{1}}{\lambda_{2}^{2}}=-\frac{(l-k)(n-(l-k)-2)}{2}\log
  q-O(1)
\end{equation*}

The next statement is the corollary from
Theorem~\ref{p:main} in the current settings.
\begin{corollary}
\label{cor:linear-flags}
  Let $(X, Y, A, B)$ be four random variables, where $(X,Y)$ is the
  uniformly random choice of a ($k$-in-$l$)--flag in
  $\Fbb_{q}^{n}$. Then
  \[
    \ing(X,Y,A,B) \geq -\frac12(l-k)(n-(l-k)-2)-O(\log H(\XYAB))
  \]
  \unskipp
\end{corollary}

\paragraph{Non-extractable mutual information.} 
Again, we compare our bound with the conventional proof of the
Ingleton inequality. If Conjecture~\ref{co:nonextractable} holds,
then
\begin{equation*}
  \begin{aligned}
    I(X; Y|W)\!+\!I(W; X|Y)\!+\!I(W; Y|X)
    &\geq
      \min\set{k(n-l)\log q,\frac{l-k}{2}\log q}\\
    &\quad
      -C\cdot\log\log q.
  \end{aligned}
\end{equation*}
In view of this inequality the MMRV inequality gives at best
\begin{equation*}
  \ing(X,Y,A,B)\geq -\frac{l-k}{2}\log q + C\cdot\log\log q
\end{equation*}
When $k=1$, this inequality is also Shannon-type, up to a
$(C\log\log q)$-term.

We now isolate the special case $k=1$ and $l=n-1$ (line-in-hyperplane
flags) of the previous corollary. The graph of
(line-in-hyperplane)-flags is a symmetric expander, and in this case
we get a cleaner bound.
\begin{corollary}
  Let $(X, Y, A, B)$ be four random variables, where $(X,Y)$ is the
  uniformly random choice of a ($1$-in-$(n-1)$)--flag in
  $\Fbb_{q}^{n}$. Then
  \[
    \ing(X,Y,A,B) \geq -O(\log H(\XYAB)).
  \]
  \unskipp
\end{corollary}

\section{Discussion}
\label{s:discussion}
\subsection{On the tightness of the bounds}
\label{s:tightness}
\subsubsection{The spectral summand}
\label{s:tightness1}
In Theorem~\ref{p:main}, the leading term in the lower bound on
$\ing(X,Y,A,B)$ is
\[
 \log \lambda_{1}(G) - 2\log \lambda_{2}(G),
\]
where graph $G=(\Xsf\sqcup \Ysf,\Esf)$ is the support of a uniform pair
$(X,Y)$.

At the same time, the lower bound on the second eigenvalue of a
biregular graph
(Proposition~\ref{p:alonlike}\eqref{eq:alonlike-strong}) implies that
\begin{equation*}
  2\log \lambda_{2}(G) \ge  \log \max\{d_{1}, d_{2}\}  -O(1),
\end{equation*}
and this bound is tight for \emph{bipartite expanders} (graphs $G$
such that $ \lambda_{2}(G) = O(\sqrt{\max\{ d_{1}, d_{2}\}})$. Note
that $\max\{d_{1}, d_{2}\} > \lambda_{1}(G)=\sqrt{d_{1}d_{2}}$ for
graphs with $d_1\neq d_2$.  This observation motivates the following
question.
\begin{question}\label{q:1}
  Can the conclusion of Theorem~\ref{p:main} be strengthened to
  the following stronger inequality?
  \begin{align*}
    \ing(X,Y,A,B)
    &\ge
      -2\log \lambda_{2}(G) + \max\{H(X| Y), H(Y|X)\}\\
    &\quad-
      O\bigl(\log H(X,Y,A,B)\bigr)
  \end{align*}
\end{question}
Note that for the bipartite expanders,  
even for non-balanced ones (including the examples discussed in the previous section), 
the expression
\[
  -2\log \lambda_{2}(G) + \max\{H(X|Y), H(Y|X)\} 
\]
reduces to $O(1)$.

\begin{remark}
  If the answer to Question~\ref{q:1} is positive, then
  Corollary~\ref{p:algflags} (for all $k>0$) and
  Corollary~\ref{cor:linear-flags} (for $k=1$ and arbitrary $l$) both
  rewrite to a stronger form
  \[
    \ing(X,Y,A,B)\geq -O(\log H(\XYAB))
  \]
  \unskipp
\end{remark}

We are not aware of any examples that rule out such a strengthening.
However, our current techniques do not seem sufficient to establish it.

\subsubsection{The $O\big(\log H(\XYAB)\big)$-term}
\label{s:tightness2}
The right-hand side of the inequality in Theorem~\ref{p:main} contains
the error term of order $\log H(\XYAB)$.  As we have observed earlier,
in the remark after Theorem~\ref{th:main}, this summand can not be
reduced to zero, in general.  This is because for the graph of
projective flags over $\Fbb_{q}$ we have
\[
  \log\frac{\lambda_{1}}{\lambda_{2}^{2}}=\log\left(1+\frac1q\right)>0
\]
Thus, without the error term, the right-hand side of the inequality in
Theorem~\ref{p:main} will be strictly positive, whereas taking $A:=X$
and $B:=Y$ makes Ingleton expression equal to zero.  

For pairs $(X,Y)$ uniformly distributed on the graph of projective flags, it is
not hard to show that the mutual information is not totally
extractable, i.e., there is no random variable $W$ satisfying
\eqref{eq:common-information}.  Let us also mention the remarkable
result\footnote{%
  We thank an anonymous referee for drawing our attention to this
  publication.}  of \cite[the second part of
Theorem~2]{csirmaz2023short}, which states that every such pair can be
extended to a quadruple $(X,Y,A,B)$ such that the Ingleton expression
$\ing(X,Y,A,B)$ is strictly negative.  However, we do not know how
negative $\ing(X,Y,\cdot,\cdot)$ can be made for pairs $(X,Y)$
uniformly distributed on a projective flag.  Can it become as small as
$-c \log H(XY)$ or even $-c \log H(XYAB)$ for some constant $c>0$, or
is it always bounded below by an absolute constant?

We stress that the lower bound provided by Theorem~\ref{th:main} 
is not uniform over all extensions of a fixed pair $(X,Y)$.
Rather, it says that very negative values of $\ing$  can occur
only when the extending variables have exponentially large entropy. We
thank the anonymous referee for this observation.

\begin{question}
  Can the term $O\big(\log H(\XYAB)\big)$ in Theorem~\ref{th:main} be
  replaced by $O(\log H(XY))$ or even $O(1)$?
\end{question}

\subsubsection{Explicit constants in the residual term}
\label{s:explict-constants}
A careful analysis of the arguments in the article suggests that the
residual term $O(\log H(\XYAB))$ in Theorem~\ref{p:main} may admit a
more explicit description. More precisely, the proofs indicate that
the error term should depend on the parameters in a form
\[
  C_{1} \cdot\log H(\XYAB)
  - C_{2}\cdot\log\bigl(1-\exp(-\epsilon_{0})  \bigr)
  + C_{3},
\]
where $C_{1},C_{2},C_{3}$ are universal constants independent of the
choice of the threshold value $\epsilon_{0}$.

However, the constants $C_{1}$ and $C_{3}$ implicitly depend on the constants
$\alpha_{4}$ and $\beta_{4}$ from Theorem~\ref{p:cool}, for which no
useful explicit bounds are currently known. 
This leads naturally to the following open problem.
\begin{problem}
  Determine the optimal constants $C_{1},C_{2},C_{3}>0$ for which an
  inequality of the following form holds: for every pair $\Gbf=(X,Y)$
  uniformly distributed on a biregular bipartite graph $\Gsf$, not
  necessarily satisfying assumptions~\eqref{eq:I(X; Y)>C}
  and~\eqref{eq:assumption}, the following bound
  holds for all extensions $(X,Y,A,B)$:
  \[
    \ing(X,Y,A,B)
    \geq
    -\min\set{
      \begin{aligned}
        &I(X; Y),\\
        &\log\frac{\lambda_{2}(\Gsf)^{2}}{\lambda_{1}(\Gsf)}
          +C_{1}\log H(\XYAB)\\
        &\qquad
          -C_{2}\log\bigl(1-\exp(-I(X; Y))\bigr)
          +C_{3}
      \end{aligned}
    },
  \]
  (if I(X; Y)=0, the minimum is interpreted as~0), 
  or establish an alternative bound of a possibly tighter form.
\end{problem}

\subsection{Spectral inequality versus MMRV inequality}
\label{s:spectral-mmrv}

We claim that for $(X,Y)$ distributed on a graph for which
$\lambda_2(\Gsf) \ll \lambda_1(\Gsf)$, the MMRV inequality is weak, while the
implication of Theorem~\ref{p:main} is rather strong. Conversely, when
$\lambda_2(\Gsf)  = \lambda_1(\Gsf)$, the MMRV inequality becomes
strong and the conclusion of Theorem~\ref{p:main} becomes trivial.
Let us explain these observations in more detail.

The spectral inequality, Theorem~\ref{p:main}, is strongest for
uniform pairs $(X,Y)$ supported on balanced expanders, in which case
it gives
\[
  \ing(X,Y,A,B)\geq -O(\log H(\XYAB)).
\]

On the other hand, if Conjecture~\ref{co:nonextractable} is true, the
MMRV inequality becomes very weak in this settings: it is Shannon-type
up to a $\log$-summand, as we demonstrate below.  Suppose $\Gbf=(X,Y)$
is a uniform pair supported on a graph $\Gsf$ satisfying
\[
  2\log\lambda_{1}=\log\lambda_{2}+O(1)
\]
where $O(1)$ term is taken in the asymptotic regime, when $H(X),H(Y)$
become large.  Then Conjecture~\ref{co:nonextractable} implies
\[
  I(X; Y|W) + I(W; X|Y) + I(W; Y|X)
    \geq
    \min\set{I(X; Y),L(X{,}Y)}-C\cdot\log H(XY).
\]
When we combine this bound with the MMRV inequality, we obtain a
rather weak claim. It gives at best the following two lower bounds:
\begin{equation*}
  \left[
    \begin{aligned}
      &\ing(X,Y,A,B)\geq -I(X; Y) + C\cdot\log H(XY)\\
      &\ing(X,Y,A,B)\geq -L(X,Y) + C\cdot\log H(XY)
    \end{aligned}
  \right.
\end{equation*}
both inequalities are Shannon-type up to the error term $C\cdot\log
H(XY)$;  see Lemma~\ref{p:shannon-ineq}.

In fact, for uniform pairs supported on arbitrary expanders, not
necessarily balanced ones, the MMRV inequality remains almost
Shannon-type, whereas Theorem~\ref{p:main} loses its strength when the
expanders lose their balance.

On the other hand, if mutual information of $\Gbf=(X,Y)$ is
extractable, namely there is an extending variable $W$ satisfying
condition~\eqref{eq:common-information}, then the MMRV inequality is a
strong claim that implies the Ingleton inequality. In this case,
unless $I(X; Y)=0$, the graph $\Gsf$ supporting the uniform pair
$\Gbf$ is disconnected, and
\[
  \lambda_{2}(\Gsf)=\lambda_{1}(\Gsf).
\]
Then the conclusion of the main theorem reads
\[
  \ing(X,Y,A,B)\geq -L(X,Y) - O(\log H(\XYAB)), 
\]
which is weaker than a Shannon-type inequality, see Lemma~\ref{p:shannon-ineq}.

\subsection{Remark on a H\"older-type expander mixing lemma}
\label{s:holder}
Our starting point in this investigation was the Expander Mixing Lemma
on the log scale, Corollary~\ref{p:mixing-log}.  Quite similarly one
could prove H\"older (rather then $\ell^{2}$) version of the lemma.

For biregular bipartite graph $G=(\Xsf\sqcup\Ysf, \Esf)$ let
\[
  M:\Rbb^{\Xsf}\times\Rbb^{\Ysf}\to\Rbb
\]
be the adjacency bilinear form. We use the infix notation,
$\phi M\psi$, for the evaluation of $M$ on $\phi\in\Rbb^{\Xsf}$ and
$\psi\in\Rbb^{\Ysf}$. Let $M'$ be the restriction of $M$ to the
orthogonal complements $\ell_{0}(\Xsf)$ and $\ell_{0}(\Ysf)$ of
constant functions in $\Rbb^{\Xsf}$ and $\Rbb^{\Ysf}$, respectively.
For $p,q\in[1,\infty]$ define the $p,q$-norm of $M'$ as
\[
  \|M'\|_{p,q}
  :=
  \sup\set{\frac{|\phi M'\psi|}{\|\phi\|_{p}\cdot\|\psi\|_{q}} \st
    0\neq\phi\in\ell_{0}(\Xsf),\; 0\neq\psi\in\ell_{0}(\Ysf)}
\]
where $\|\phi\|_{p}$ is the $\ell^{p}$-norm of the function $\phi$
with respect to the counting measure on $\Xsf$.
\begin{proposition}\label{p:mixing-log-hölder}
  For any subgraph $W$ in a biregular bipartite graph
  $G=(\Xsf\sqcup\Ysf, \Esf)$ and any $p,q\in[1,\infty]$ at least one
  of the following inequalities holds:
  \begin{enumerate}[label=(\roman*)]
  \item\label{i:boring-hölder}
    $\log |\Esf_{W}|-\log |\Xsf_{W}|-\log |\Ysf_{W}| \leq \log |\Esf| - \log |\Xsf| -
    \log |\Ysf| + \log2$
  \item[or]
  \item\label{i:interesting-hölder}
    $\log |\Esf_{W}|-\frac1p\log |\Xsf_{W}|-\frac1q\log| \Ysf_{W}|
    \leq
    \log\|M'\|_{p,q}+\log 8$
  \end{enumerate}
  \unskipp
\end{proposition}
We then can proceed in a similar manner as in
Sections~\ref{s:ing-q-unif}--\ref{s:main}. However, for general $p,q$
the norm $\|M'\|_{p,q}$ is very difficult to estimate even in the
simplest of examples, and at the moment we do not know of any
situations, where it would be useful. For this reason we did not
pursue this line of research in the present article.

\section*{Acknowledgments}
The ideas leading to this work originated during a short visit of the
second author and Jim Portegies (TU Eindhoven) to the first author at
the Max Planck Institute for Mathematics in the Sciences, Leipzig, in
2020.  Both authors thank the Institute for its hospitality and Jim
Portegies for valuable discussions.
\smallskip

\noindent
The authors are grateful to the anonymous reviewers of the journal IEEE Transactions on Information theory for their careful
reading of the manuscript and for their insightful and constructive comments, which led to substantial improvements in the presentation
and exposition of the paper.

\smallskip

\noindent
This work was supported by the Agence Nationale de la Recherche (ANR)
under project ANR-21-CE48-0023.

\printbibliography[heading=bibintoc]

\appendix

\section{Shannon-type inequalities used in the article}
\label{s:shannon}
\def\thelemma{\thesection.\Alph{lemma}}
\begin{lemma}
  \label{p:shannon-ineq} 
Let $(X,Y,A,B,W)$ be a tuple of jointly distributed random
variables. Then the following inequalities are Shannon-type:
\begin{enumerate}[label=(\roman*)]
\item\label{i:abw} $I(A; B)\geq H(W)-H(W|A)-H(W|B)$;

\item\label{i:x:y:w}
$H(W) -  I(X;Y)  =  H(W|X) + H(W|Y) - H(W|XY) - I(X;Y|W)$;

\item\label{i:ing-i} $\ing(X,Y,A,B)\geq -I(X; Y)$; 
\item\label{i:ing-h} $\ing(X,Y,A,B)\geq -H(X|Y)$; 
\item\label{i:ing-hh/2} $\ing(X,Y,A,B)\geq-\L(X,Y)$; 
\item\label{i:ing-lll} $\ing(X,Y,A,B)\geq -\L(X,Y|A)-\L(X,Y|B)+\L(X,Y)$; 

\item\label{i:ing-soft}
  $\displaystyle
  \begin{aligned}[t]
    \ing(X,Y,A,B)\geq&{}   - I(X; Y|W) -  H(W|X) -  H(W|Y)  - H(W|XY)\\ 
    \geq&  {} - I(X; Y|W) -  H(W|X) -  H(W|Y)  
     - \min \{ H(W|X), H(W|Y)\}.
  \end{aligned}
  $
\end{enumerate}
\end{lemma}
\begin{proof}
  All \emph{equalities} of entropic expressions below can be verified
  by expanding both left- and right-hand sides in absolute entropies
  of the joints. Reader is also encouraged to draw Venn diagrams
  explained in \cite{yeung2012first} as a guiding principle.
   
 \begin{enumerate}[label=(\roman*)]
  \item $I(A; B)-\big(H(W)-H(W|A)-H(W|B)\big)
    =
      I(A; B|W)+H(W|AB)\geq0
      $; 

\item It is enough to use the standard expansions for conditional
  entropy and mutual information,
  \[
    \begin{array}{rcl}
      H(W|X) &=& H(XW) - H(X), \\
      H(W|Y) &=&  H(YW) - H(Y),\\
      H(W|XY) &=&  H(XYW)- H(XY),\\
      I(X;Y) &=&  H(X) + H(Y) - H(XY),\\
      I(X;Y|W) &=&  H(XW) + H(YW) - H(XYW) - H(W) 
    \end{array}
  \]
  and substitute them in \ref{i:x:y:w}. 
  
\item $\ing(X,Y,A,B)=I(X; Y|A)+I(X; Y|B)+I(A; B)-I(X; Y)\geq-I(X; Y)$; 
  
\item $\ing(X,Y,A,B)+H(X|Y)=H(X|Y\!AB)+I(X; Y|AB)+I(X; A|Y)$\\
  $\mbox{}\mkern220mu + I(X; B|Y) + I(A; B|X)\geq0$; 
\item Take inequality in \ref{i:ing-h} and the symmetric one
  obtained by transposing $X$ and $Y$. Their average gives the
  required inequality. 
\item Apply inequality \ref{i:abw} with $W:=X$
  and with $W:=Y$ and average:
  \begin{equation*}
    \begin{aligned}
      I(X;Y) \leq&\ I(X; Y|A) + I(X; Y|B) + I(A; B)\\
                 &  + I(X;Y) - H(W)+ 2H(W|X) + 2H(W|Y).
    \end{aligned}
  \end{equation*}
  Substitute the above inequality in the expression for
  $\ing(X,Y,A,B)$:
  \begin{equation*}
    \begin{aligned}
      \ing(X,Y,A,B)
      &:=
        I(X; Y|A) + I(X; Y|B) + I(A; B) -
        I(X; Y)\\
      &\geq\quad
        I(X; Y|A) - \frac12 H(X|A) -\frac12 H(Y|A)\\
      &\quad+
        I(X; Y|B) - \frac12 H(X|B) -\frac12 H(Y|B)\\
      &\quad+
        \frac12 H(X) +\frac12 H(Y)-I(X; Y)\\
      &=-\L(X,Y|A)-\L(X,Y|B)+\L(X,Y)
    \end{aligned}
  \end{equation*}
\item We use inequality~\ref{i:abw} 
  \[
    H(W) \le H(W|A) + H(W|B) + I(A; B)
  \]
  and two ``conditional'' instances of a similar inequality
  $ H(W) \le H(W|X) + H(W|Y) + I(X; Y)$,
  \begin{equation*}
    \begin{aligned}
      H(W|A) &\leq H(W|X,A) + H(W|Y,A) + I(X; Y|A)\\
             &\leq H(W|X) + H(W|Y) + I(X; Y|A) 
    \end{aligned}
  \end{equation*}
  and   
  \begin{equation*}
    \begin{aligned}
      H(W|B) &\leq H(W|X,B) + H(W|Y,B) + I(X; Y|B)\\
             &\leq H(W|X) + H(W|Y) + I(X; Y|B). 
    \end{aligned}
  \end{equation*}
  The sum of these three inequalities gives
  \begin{equation*}
    \begin{aligned}
      H(W)
      &\leq I(X; Y|A) + I(X; Y|B) + I(A; B)\\
      &\quad
        + 2H(W|X) + 2H(W|Y),
    \end{aligned}
  \end{equation*}
  which rewrites to
  \begin{equation*}
    \begin{aligned}
      I(X;Y)
      &\leq
        I(X; Y|A) + I(X; Y|B) + I(A; B)\\
      &\quad
        + I(X;Y) - H(W)+ 2H(W|X) + 2H(W|Y)\\
    \end{aligned}
  \end{equation*}
  We combine this inequality with \ref{i:x:y:w} and obtain 
  \begin{equation*}
    \begin{aligned}
      I(X;Y) \leq&
                   \ I(X; Y|A) + I(X; Y|B) + I(A; B)\\
                 & -  H(W|X) - H(W|Y) + H(W|XY) + I(X;Y|W)\\
                 &+ 2H(W|X) + 2H(W|Y)\\
      =&\ I(X; Y|A) + I(X; Y|B) + I(A; B)\\
                 &+ I(X;Y|W) +  H(W|X) + H(W|Y) + H(W|XY) .
    \end{aligned}
  \end{equation*}
  It remains to observe that
  \[
    H(W|XY) \le \min\{H(W|X), H(W|Y)  \} .
  \]
  
\end{enumerate}
\end{proof}

\begin{remark}
  Elementary manipulations with entropy inequalities, and in
  particular checking whether a given inequality is of Shannon-type,
  can be carried out mechanically with the help of a computer.
  Several open-source software packages are available for this
  purpose. The general theory behind such solvers (the reduction of an
  information-theoretic problem to a linear programming problem) is
  explained in~\cite{yeung2012first,yeung2021machine}.  The first
  software package for checking whether an information inequality is
  of Shannon-type was a MATLAB-based solver ITIP~\cite{yeung1996itip}
  (see \cite{yeung1997framework} for the theoretical framework behind
  this software).  Several further developments of this approach
  appeared later:
\begin{itemize}
\item Xitip~\cite{pulikkoonattu2008xitip}, a platform-independent version of ITIP; 
\item  MINITIP~\cite{minitip}, a descendant of ITIP, written in C, uses GLPK (GNU Linear Programming Kit) as the LP solver; 
\item AITIP~\cite{pulikkoonattu2020aitip} can produce human-readable proofs and attempts to minimize the number of proof steps;  it is based on the theory presented in~\cite{ho2020proving}; 
\item PSITIP~\cite{PSITIP},  a more advanced computer algebra system for information theory written in Python;  it is based on the framework introduced in~\cite{li2023automated}.
\end{itemize}
A more detailed discussion of these software packages can be found in~\cite{yeung2021machine}.
\end{remark}

\end{document}